\documentclass[12pt,letterpaper,draftcls,oneside,peerreview,onecolumn]{IEEEtran}
\usepackage{amsmath}
\usepackage{amsfonts}
\usepackage{amsmath}
\usepackage{graphicx}
\usepackage{amssymb}
\usepackage{cite}
\usepackage{balance}
\usepackage[left=0.9in,top=1.0in,bottom=1.2in,right=0.9in]{geometry}
\usepackage[OT1,T1]{fontenc}

\newtheorem{theorem}{Theorem}

\newtheorem{corollary}{Corollary}

\newtheorem{lemma}{Lemma}

\newtheorem{proposition}{Proposition}

\newenvironment{proof}[1][Proof]{\noindent\textbf{#1.} }{\ \rule{0.5em}{0.5em}}
\allowdisplaybreaks

\begin{document}

\title{Relay Selection for Simultaneous Information Transmission and
Wireless Energy Transfer: \\
A Tradeoff Perspective}
\author{Diomidis S. Michalopoulos,~\IEEEmembership{Member,~IEEE,} Himal A.
Suraweera,~\IEEEmembership{Member,~IEEE}, and Robert Schober,~%
\IEEEmembership{Fellow,~IEEE} \thanks{%
D. S. Michalopoulos and R. Schober are with the Department of
Electrical and Computer Engineering, University of
Erlangen-Nuremberg, Germany, (e-mail:
\{michalopoulos, schober\}@lnt.de)} \thanks{%
H. A. Suraweera is with the Singapore University of Technology and
Design, 20 Dover Drive, Singapore 138682 (e-mail:
himalsuraweera@sutd.edu.sg).}
 \vspace{-0.5cm}} \maketitle

\vspace{-0.5cm}
\begin{abstract}
In certain applications, relay terminals can be employed to simultaneously
deliver information and energy to a designated receiver and a radio
frequency (RF) energy harvester, respectively. In such scenarios, the relay
that is preferable for information transmission does not necessarily
coincide with the relay with the strongest channel to the energy harvester,
since the corresponding channels fade independently. Relay selection thus
entails a tradeoff between the efficiency of the information transfer to the
receiver and the amount of energy transferred to the energy harvester. The
study of this tradeoff is the subject on which this work mainly focuses.
Specifically, we investigate the behavior of the ergodic capacity and the
outage probability of the information transmission to the receiver, for a
given amount of energy transferred to the RF energy harvester. We propose
two relay selection methods that apply to any number of available relays.
Furthermore, for the case of two relays, we develop the optimal relay
selection method in a maximum capacity / minimum outage probability sense,
for a given energy transfer constraint. A close-to-optimal selection method
that is easier to analyze than the optimal one is also examined. Closed-form
expressions for the capacity-energy and the outage-energy tradeoffs of the
developed schemes are provided and corroborated by simulations. Interesting
insights on the aforementioned tradeoffs are obtained.
\end{abstract}


\newpage

\section{Introduction}

The approach towards energy consumption in communication systems has
experienced a drastic change in the last few years. The enormous growth of
telecommunication networks has lead to a massive increase of their energy
consumption. Forecasts on the energy consumption of future applications
place information and communication technology (ICT) networks among the big
energy consumers, so that the energy consumed by ICT infrastructure
worldwide is anticipated to reach the current level of the total global
electricity consumption in the next 20-25 years \cite{C:FettweisICT,ST:Cisco}%
. Hence, owing to this growing concern regarding the energy footprint of
communications, modern architectures consider energy not as an unlimited
resource, as it traditionally was, but as a scarce resource which plays a
significant role in system design \cite{B:Hoss_Green,J:Zia}.

In line with the contemporary trend towards renewable sources, energy
harvesting appears as a viable solution to powering wireless communications
nodes \cite{J:Medepally,J:Ho1,J:Ozel}. In addition, energy harvesting offers
wireless communication substantial flexibility, since wireless nodes are not
necessarily attached to a fixed power supply nor are they dependent on
battery replacement and/or recharge \cite{J:Ho1}. The most common forms of
harvested energy used in wireless communications are solar energy,
piezoelectric energy, and energy harvested from radio frequency (RF)
transmissions \cite{J:Sudeva,J:Nasir,J:HuangLau}. The latter form attracts
particular interest as it allows terminals with low energy requirements to
be remotely powered, thereby it provides a feasible solution for cases where
remote energy supply is the only powering option (for example, in body area
networks where devices are implanted into the human body such that accessing
them is impossible) \cite{ST:Pcast}. Moreover, as information and energy are
transmitted via the same signal, RF energy harvesting poses challenges on
the efficient design of systems that provide simultaneous information and
energy transfer to the same terminal \cite%
{J:KrikidisRF,C:Varshney,C:ShanTes,C:Shen} or to different terminals \cite%
{J:ZhangMIMO,C:Chalise,C:IshibashiTarokh}.

\emph{Motivation: }Relay-assisted communication and particularly relay
selection offers a substantial improvement to the quality of service in
wireless networks, particularly in scenarios where source and destination
are located far apart from one another \cite{B:Mischa_Coop,B:Fitzek_Coop}.
Naturally, the same concept applies also to wireless energy transfer
scenarios, as the large path-loss of the energy-bearing channel renders
wireless energy transfer over large distances prohibitive. In this regard,
suppose that relay terminals are used both for assisting the information
transmission to a designated receiver and the energy transfer to a
designated RF energy harvester, which are located far from each other. Then,
the question that arises is which relay to activate in each transmission
session, as the activated relay will provide the receiver and the harvester
with data and energy respectively and the channels used for information and
energy transfer vary independently from each other. Since the relay that
provides the most efficient data transmission to the receiver does not
necessarily coincide with the relay that provides the largest energy
transfer to the RF harvester, a tradeoff is revealed: The quality of the
information transmission to the receiver is exchanged for the efficiency of
the energy transfer to the harvester. This tradeoff is reflected in the
decision on which relay is selected, and represents the main topic of
interest of this work.

\emph{Contribution:} We show via mathematical and numerical analysis that,
depending on the application scenario and the available amount of channel
state information (CSI), the achievable tradeoff can range from a linear
exchange between data and energy transmission to the optimal feasible
tradeoff. In particular, our results can be summarized as follows. For the
versatile scenario where $N$ relays are available for information forwarding
and wireless energy transfer, we study two relay selection schemes, namely
the time-sharing and the threshold-checking scheme, in terms of the
achievable tradeoff between average energy transfer and ergodic capacity, as
well as the tradeoff between energy transfer and outage probability. For the
case where two relays are available ($N=2$), we propose a selection method
that attains the optimal achievable tradeoff (i.e., the optimal ergodic
capacity and/or outage probability for any given energy transfer), along
with a similar selection method which behaves approximately as the optimal
one in certain regions. Nevertheless, as both of these selection methods
require global CSI knowledge in each transmission session, the time-sharing
and the threshold-checking schemes are of interest in scenarios with limited
CSI availability.

\emph{Organization: }Useful insights regarding the tradeoff between the
information transmission to the receiver and the energy transfer to the RF
harvester are provided in Section \ref{NUM}, where an extensive discussion
on the derived results is given. The tradeoff results pertain to the optimal
schemes which are developed in Section \ref{PAR} for $N=2$, as well as the
versatile schemes of time-sharing and threshold-checking which apply to any
number of relays. These schemes are presented in detail in Section \ref{PRE}%
, and later analyzed in terms of the achievable tradeoff between
ergodic capacity and energy transfer (Section \ref{TRAD}) and the
outage performance for a given energy transfer constraint (Section
\ref{OUT}).
Prior to the analysis, the preliminaries of the considered
system model and some fundamental tradeoff features are presented
in the ensuing, Section \ref{PRE}.

\section{\label{PRE}Preliminaries}

\subsection{System Model}

Sketched in Fig. \ref{Descr}, the considered setup simultaneously transfers
information from a source terminal, $\mathcal{S}$, to a destination
terminal, $\mathcal{D}$, and energy from $\mathcal{S}$ to a harvester
terminal, $\mathcal{H}$. Both the information and energy transferring
processes are assisted by a set of half-duplex decode and forward (DF)
relays, denoted by $\mathcal{R}_{i}$, $i=1,...,N$. All terminals are assumed
to be equipped with a single antenna. The information transmission to $%
\mathcal{D}$ and the energy transfer to $\mathcal{H}$ take place via one of
the relays, based on a process described in Subsection \ref{RelSel}.

Let $h_{AB}$ denote the channel between terminals $A$ and $B$, where $%
A\in\left\{ \mathcal{S},\mathcal{R}_{1},...,\mathcal{R}_{N}\right\} $ and $%
B\in\left\{ \mathcal{R}_{1},...,\mathcal{R}_{N},\mathcal{D},\mathcal{H}%
\right\} $. Let us denote the squared channel gain of the $A$-$B$ link by $%
a_{AB}=\left\vert h_{AB}\right\vert ^{2}$. For simplicity, we assume that
the fading in all channels involved is Rayleigh, independent and identically
distributed (i.i.d.). However, the analysis can be extended to account for
independent but not necessarily identically distributed fading channels as
well. The instantaneous signal-to-noise ratio (SNR) of the $A$-$B$ link is
denoted by $\gamma_{AB}$, and is exponentially distributed with mean value $%
\bar{\gamma}$. Moreover, the source and the relay terminals are assumed to
transmit with power $P$.

\begin{figure}[ptb]
\centering
\includegraphics[trim=1.2cm 17.1cm 4.5cm 2.8cm, clip=true,
scale=0.55]{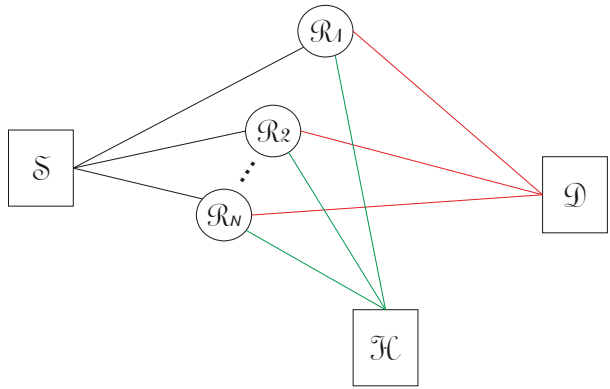}
\caption{The considered setup. Red and black lines indicate information
transfer; green lines indicate energy transfer.}
\label{Descr}
\end{figure}

In DF relaying, the composite $\mathcal{S}$-$\mathcal{R}_{i}$-$\mathcal{D}$
path is dominated by the \textquotedblleft bottleneck\textquotedblright~link
(see e.g. \cite{J:Behrouz1}). Hence, the equivalent SNR, $\gamma_{i}$, of
the $\mathcal{S}$-$\mathcal{R}_{i}$-$\mathcal{D}$ link, is defined as%
\begin{equation}
\gamma_{i}=\min\left( \gamma_{\mathcal{SR}_{i}},\gamma_{\mathcal{R}_{i}%
\mathcal{D}}\right) .  \label{min}
\end{equation}
Because $\gamma_{\mathcal{SR}_{i}}$ and $\gamma_{\mathcal{R}_{i}\mathcal{D}}$
are independent exponentially distributed random variables (RVs), $%
\gamma_{i} $ is also exponentially distributed and its mean value equals
half of the mean value of $\gamma_{\mathcal{SR}_{i}}$ and $\gamma_{\mathcal{R%
}_{i}\mathcal{D}}$, i.e.%
\begin{equation}
f_{\gamma_{i}}\left( x\right) =\frac{2}{\bar{\gamma}}\exp\left( -\frac {2x}{%
\bar{\gamma}}\right) .  \label{gi}
\end{equation}
We denote by $\varepsilon_{i}$ the energy transferred to $\mathcal{H}$ via
the $\mathcal{S}$-$\mathcal{R}_{i}$-$\mathcal{H}$ path, i.e., the harvested
energy when $\mathcal{R}_{i}$ is selected. This energy is given by%
\begin{equation}
\varepsilon_{i}=\beta~P~a_{\mathcal{R}_{i}\mathcal{H}}
\end{equation}
where $\beta$, $0<\beta\leq1$, denotes the energy absorption
coefficient, which equals the energy absorbed by $\mathcal{H}$
when the received power at $\mathcal{H}$ equals one. Roughly
speaking, parameter $\beta$ characterizes the efficiency of the
energy harvester \cite{J:ZhangMIMO}. The noise power is assumed
identical in all links, and denoted by $N_{0}$. Moreover, we
assume that because of large path-loss and/or shadowing no
information and energy
are transferred via the $\mathcal{S}$-$\mathcal{D}$ and $\mathcal{S}$-$%
\mathcal{H}$ channels, respectively.

\subsection{\label{RelSel}Relay Selection: General Description}

In each transmission frame, a single relay out of the set of available
relays is selected. The selected relay is denoted by $\mathcal{R}_{s}$: That
is, $s=i$ if $\mathcal{R}_{i}$ is selected, $i=1,...,N$. The selection is
assumed to be implemented in a centralized manner. That is, a central unit
(CU) collects the CSI of all the links in the system. Based on the collected
CSI, the CU decides which relay should be selected for a given transmission
frame. Loosely speaking, the decision on the selected relay tries to
compromise between the reliability of the information transmission to $%
\mathcal{D}$ and the total energy transferred to $\mathcal{H}$.

Let $\mathcal{R}_{\kappa}$ denote the relay which maximizes the SNR at $%
\mathcal{D}$ at a given transmission frame. That is,%
\begin{equation}
\kappa=\arg\max_{i=1,...,N}\gamma_{i}.
\end{equation}
Let $\mathcal{R}_{\lambda}$ denote the relay which maximizes the energy
transfer to $\mathcal{H}$ at a given transmission frame. That is,%
\begin{equation}
\lambda=\arg\max_{i=1,...,N}\varepsilon_{i}.  \label{lamda}
\end{equation}
Apparently, as $\mathcal{R}_{\kappa}$ and $\mathcal{R}_{\lambda}$ are not
necessarily identical, the selection of $\mathcal{R}_{s}$ leads to a
tradeoff between information transmission and energy transfer, which is
analyzed in detail in Section \ref{TRAD}. Prior to elaborating on the
particular tradeoff of interest, some preliminaries on the tradeoff analysis
are in order.

\subsection{Preliminaries of Tradeoff Analysis}

In economics and several fields of engineering, a tradeoff is referred to as
a situation where one commodity or performance metric, $X$, is sacrificed in
return for gaining another commodity or performance metric, $Y$ \cite%
{B:Starr}. The tradeoff is usually illustrated by a 2-dimensional curve
which consists of the set of all feasible $\left( X,Y\right) $ pairs. An
illustrative example of tradeoff curves is presented in Fig. \ref{TR}, where
the range of the exchanged metrics is normalized to one.

For a better understanding of the subsequent analysis, the following
terminology is introduced:

\begin{itemize}
\item \emph{Tradeoff factor:} By tradeoff factor, $\delta$ ($0\leq\delta\leq
1$), we refer to the priority of maximizing $X$ over $Y$. The tradeoff
factor specifies the point of operation along the tradeoff curve, i.e., it
specifies the quantity of $X$ that is exchanged with $Y$. When the range of $%
X$ is normalized to the interval between zero and one (c.f. Fig. \ref{TR}),
the tradeoff factor equals the abscissa of the point of interest.

\begin{figure}[ptb]
\centering
\includegraphics[trim=1.4cm 14.9cm 8.5cm 6.8cm, clip=true,
scale=1.1]{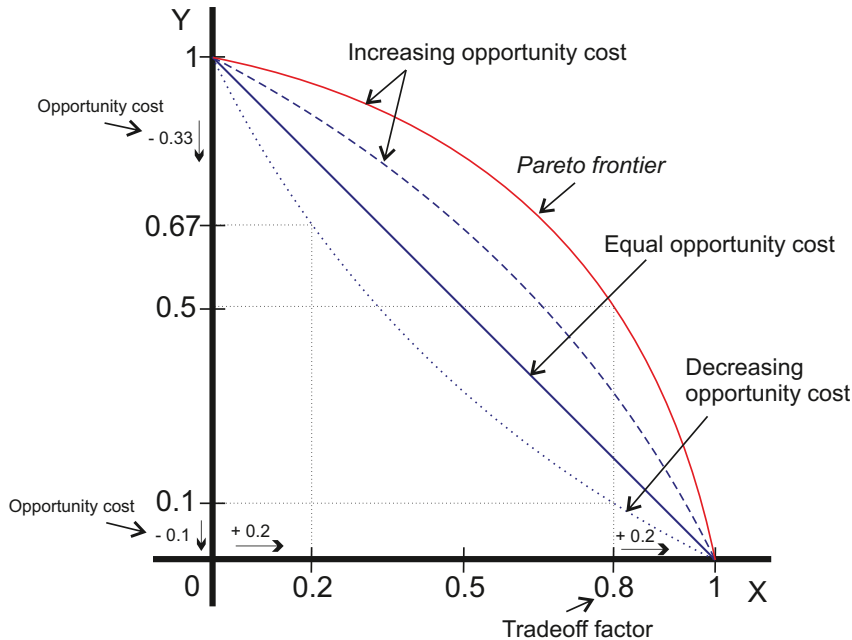}
\caption{General tradeoff parameters between two variables, $X$, $Y$.}
\label{TR}
\end{figure}

\item \emph{Opportunity cost: }For a given difference of $X$, the
opportunity cost is defined as the corresponding absolute difference of $Y$
\cite{B:Starr}. A linear tradeoff curve thus corresponds to an
\textquotedblleft\textit{equal opportunity cost}\textquotedblright\ tradeoff
(c.f. Fig. \ref{TR}, solid blue line), since it entails a constant exchange
between $X$ and $Y$. A strictly convex tradeoff curve is associated with a
tradeoff with \textquotedblleft\textit{decreasing opportunity cost}%
\textquotedblright\ (c.f. Fig. \ref{TR}, dotted blue line) while a strictly
concave curve corresponds to a tradeoff with \textquotedblleft\textit{%
increasing opportunity cost}\textquotedblright\ (c.f. Fig. \ref{TR}, dashed
blue line). In practice, tradeoffs with increasing opportunity cost are
preferable because they entail a relatively larger gain of one of the two
metrics for a given sacrifice of the other.

\item \emph{Pareto frontier: }An $\left( X,Y\right) $ allocation is
considered \textquotedblleft Pareto efficient\textquotedblright\ if there
exist no other feasible $\left( X,Y\right) $ allocation which results in
increasing one metric ($X$ or $Y$) without decreasing the other \cite%
{B:Starr}. The set of all Pareto efficient points comprises the Pareto
frontier (c.f. Fig. \ref{TR}, red line). The Pareto frontier illustrates the
optimal tradeoff between $X$ and $Y$, in the sense that it provides the
largest achievable value of $Y$ ($X$) given $X$ ($Y$).
\end{itemize}

\section{\label{TRAD}Tradeoff Between Ergodic Capacity and Average
Transferred Energy}

This section presents an analysis of the tradeoff between the ergodic
capacity for information transmission to $\mathcal{D}$ and the average
energy transfer to $\mathcal{H}$. In the sequel, we consider the energy
transfer as the reference metric ($X$) and the ergodic capacity as the cost
metric ($Y$). This tradeoff is governed by the decision regarding the relay
selection. That is, since $\mathcal{R}_{\kappa }$ and $\mathcal{R}_{\lambda
} $ do not necessarily coincide with each other, the choice of the activated
relay determines how much of the available capacity for information
transmission to $\mathcal{D}$ is exchanged for energy transfer to $\mathcal{H%
}$, or, in other words, the tradeoff factor. Here, we consider three relay
selection schemes, which we dub \textquotedblleft
time-sharing\textquotedblright , \textquotedblleft
threshold-checking\textquotedblright , and \textquotedblleft weighted
difference\textquotedblright\ schemes, respectively. The three considered
schemes have different CSI requirements, thus depending on the CSI
availability they can be employed in different application scenarios. The
resulting tradeoffs are investigated in Subsections \ref{TS0}, \ref{TB0},
and \ref{WD0}, respectively, while a discussion on their implementation
complexity is given in Subsection \ref{IC0}. Prior to analyzing the specific
tradeoffs enabled by the considered schemes, we first study their
boundaries, namely the minimum and maximum achievable ergodic capacity and
energy transfer. These boundaries are identical for all three schemes.

\subsection{\label{BOUND}Tradeoff Boundaries}

\subsubsection{Ergodic capacity}

\begin{lemma}
\label{lem1}The minimum and maximum ergodic capacity for
information
transmission to $\mathcal{D}$ equals, respectively,%
\begin{align}
C_{\min } &=\frac{\exp \left( \frac{2}{\bar{\gamma}}\right) \mathcal{E}%
_{1}\left( \frac{2}{\bar{\gamma}}\right) }{2\ln \left( 2\right) }
\label{C1} \\
C_{\max } &=N\sum_{j=0}^{N-1}\frac{\left( -1\right) ^{j}\binom{N-1}{j}}{%
2\left( j+1\right) \ln \left( 2\right) }\exp \left( 2\frac{j+1}{\bar{\gamma}}%
\right) \mathcal{E}_{1}\left( 2\frac{j+1}{\bar{\gamma}}\right)
\label{CN}
\end{align}%
where $\mathcal{E}_{n}\left( x\right) =\int_{1}^{\infty }e^{-xy}/\left(
y^{n}\right) dy$ is the exponential integral function \cite[Eq. (5.1.1)]%
{B:Abr_Ste_Book}.
\end{lemma}

\begin{proof}
The proof is provided in Appendix \ref{Cmax}.
\end{proof}

\subsubsection{Average Transferred Energy}

\begin{lemma}
\label{lem2}The minimum and maximum average energy transfer to
$\mathcal{H}$
equals, respectively,%
\begin{align}
\epsilon_{\min}&=\bar{\varepsilon}  \label{emin}\\
\epsilon_{\max}&=H_{N}\bar{\varepsilon}  \label{emax}
\end{align}
where $\bar{\varepsilon}=\beta N_{0}\bar{\gamma}$ is the expectation of $%
\varepsilon_{i}$ and $H_{N}$ denotes the harmonic number of $N$ defined as $%
H_{N}=\sum_{i=1}^{N}\left( 1/i\right) $.
\end{lemma}

\begin{proof}
Since the maximum energy transfer occurs for $s=\lambda$, (\ref{emax}) is
obtained as the first order moment of $N$ i.i.d. exponentially distributed
RVs \cite{B:David}. Eq. (\ref{emin}) is trivially obtained by assuming that
the relay selection process is independent of $a_{\mathcal{R}_{i}\mathcal{H}%
} $, $i=1,...,N$.
\end{proof}

\subsubsection{Tradeoff factor}

Let $\epsilon$ denote the average (long-term) energy harvested by $\mathcal{H%
}$. Out of the two exchanged metrics (i.e., capacity $C$ and energy transfer
$\epsilon$), $\epsilon$ is considered as the reference metric. Hence, the
tradeoff factor, $\delta$ ($0\leq\delta\leq1$), is the percentage of the
range of possible energy transfer that is actually transferred to $\mathcal{H%
}$. Considering the boundaries of $\epsilon$ given in (\ref{emax}), (\ref%
{emin}), we can express $\epsilon$ in terms of the tradeoff factor as $%
\epsilon=\left[ 1+\delta\left( H_{N}-1\right) \right] \bar{\varepsilon}$.
The tradeoff factor is thus obtained by solving this expression with respect
to $\delta$, yielding

\begin{equation}
\delta=\frac{\frac{\epsilon}{\bar{\varepsilon}}-1}{H_{N}-1}.  \label{delta}
\end{equation}

\subsection{\label{TS0}Time-Sharing Selection Scheme}

The time-sharing scheme is considered as a simplistic selection method which
operates as follows: In each transmission frame, the CU selects either $%
\mathcal{R}_{\kappa }$ or $\mathcal{R}_{\lambda }$ in a pseudorandom
fashion. That is, $\mathcal{R}_{\kappa }$ is selected with probability $\mu $%
; $\mathcal{R}_{\lambda }$ is selected with probability $1-\mu $, i.e.%
\begin{equation}
s=\left\{
\begin{array}{c}
\kappa ,\text{ with probability }\mu \\
\lambda ,\text{ \ with probability }1-\mu%
\end{array}%
\right. .  \label{TS1}
\end{equation}%
This strategy ensures that, in the long run, the percentage of transmission
frames allocated for optimum information transmission and optimum energy
transfer is controllable.

The capacity of the time-sharing scheme equals $C_{\max}$ for the time
frames when $s=\kappa$, and $C_{\min}$ for the time frames where $s=\lambda$%
. Consequently, the ergodic capacity is obtained as%
\begin{equation}
C_{TS}=\mu C_{\max}+\left( 1-\mu\right) C_{\min}.  \label{CPS}
\end{equation}
Similarly, the energy transfer to $\mathcal{H}$ equals $\epsilon_{\min}$ if $%
s=\kappa$ and $\epsilon_{\max}$ if $s=\lambda$. Hence, the average energy
transferred to $\mathcal{H}$ when the time-sharing selection method is
employed is given by%
\begin{equation}
\epsilon_{TS}=\mu\bar{\varepsilon}+\left( 1-\mu\right) \bar{\varepsilon }%
~H_{N}=\bar{\varepsilon}~\left[ \mu+\left( 1-\mu\right) ~H_{N}\right] .
\label{eps1}
\end{equation}
Solving (\ref{eps1}) with respect to $\mu$ yields%
\begin{equation}
\mu=\frac{\bar{\varepsilon}H_{N}-\epsilon_{TS}}{\bar{\varepsilon}\left(
H_{N}-1\right) }.  \label{mu}
\end{equation}
By plugging (\ref{mu}) into (\ref{CPS}), $C_{TS}$ is expressed as a function
of $\epsilon_{TS}$ as follows
\begin{equation}
C_{TS}=\frac{e^{\frac{2}{\bar{\gamma}}}\left( \epsilon_{TS}-\bar{\varepsilon
}\right) \mathcal{E}_{1}\left( \frac{2}{\bar{\gamma}}\right) +\left( \bar{%
\varepsilon}~H_{N}-\epsilon_{TS}\right) \ln\left( 4\right) \sum _{j=0}^{N-1}%
\frac{N\left( -1\right) ^{j}e^{\frac{2\left( j+1\right) }{\bar{\gamma}}%
\binom{N-1}{j}}\mathcal{E}_{1}\left( \frac{2\left( j+1\right) }{\bar{\gamma}}%
\right) }{2\left( j+1\right) \ln\left( 2\right) }}{2\bar{\varepsilon}~\left(
H_{N}-1\right) \ln\left( 2\right) }.  \label{CTS2}
\end{equation}

\begin{corollary}
The time-sharing scheme results in an equal opportunity cost between ergodic
capacity and average energy transfer.
\end{corollary}

\begin{proof}
The proof follows directly from (\ref{CTS2}), by noting that $C_{TS}$ is a
linear function of $\epsilon_{TS}$.
\end{proof}

\subsection{\label{TB0}Threshold-Checking Selection Scheme}

The time-sharing scheme gives insight into the tradeoff between information
transmission and energy transfer, yet, clearly, it is far from making full
use of the relays. Thus, we consider an alternative selection scheme, which
operates as follows. In each transmission session the SNR of the $\mathcal{S}
$-$\mathcal{R}_{\kappa }$-$\mathcal{D}$ link, $\gamma _{\kappa }$, is
compared to a threshold, $\tau $. If $\gamma _{\kappa }\geq \tau $, then $%
\mathcal{R}_{\kappa }$ is activated. Otherwise, $\mathcal{R}_{\lambda }$ is
activated. In mathematical terms,%
\begin{equation}
s=\left\{
\begin{array}{c}
\kappa ,\text{ if }\gamma _{\kappa }\geq \tau \\
\lambda ,\text{ \ if }\gamma _{\kappa }<\tau%
\end{array}%
\right. .  \label{TB1}
\end{equation}%
Naturally, this selection scheme is expected to lead to a better
capacity-energy transfer tradeoff. A rigorous analysis of this tradeoff
follows.

Considering the threshold-checking scheme's mode of operation, the ergodic
capacity of the information transfer to $\mathcal{D}$ is obtained as (for
the derivation, please refer to Appendix \ref{ApCTB})%
\begin{align}
C_{TC}& =E\left\{ C_{\kappa }\left\vert \gamma _{\kappa }\geq \tau \right.
\right\} \Pr \left\{ \gamma _{\kappa }\geq \tau \right\} +E\left\{
C_{\lambda }\left\vert \gamma _{\kappa }<\tau \right. \right\} \Pr \left\{
\gamma _{\kappa }<\tau \right\}  \notag \\
& =\sum_{j=0}^{M-1}\frac{N\left( -1\right) ^{j}\binom{N-1}{j}e^{-\frac{%
2\left( j+1\right) \tau }{\bar{\gamma}}}\left[ e^{\frac{2\left( j+1\right)
\left( \tau +1\right) }{\bar{\gamma}}}\mathcal{E}_{1}\left( \frac{2\left(
j+1\right) \left( \tau +1\right) }{\bar{\gamma}}\right) +\ln \left( \tau
+1\right) \right] }{2\left( j+1\right) \ln \left( 2\right) }  \notag \\
& +\frac{e^{\frac{2}{\bar{\gamma}}}\left[ \mathcal{E}_{1}\left( \frac{2}{%
\bar{\gamma}}\right) -\mathcal{E}_{1}\left( \frac{2\left( 1+\tau \right) }{%
\bar{\gamma}}\right) \right] -e^{-\frac{2\tau }{\bar{\gamma}}}\ln \left(
1+\tau \right) }{2\ln \left( 2\right) }\left( 1-e^{-\frac{2\tau }{\bar{\gamma%
}}}\right) ^{N-1}  \label{CTB}
\end{align}%
where $E\left\{ \cdot \right\} $ denotes expectation. Similarly to (\ref%
{eps1}), the average energy transfer to $\mathcal{H}$ is obtained as%
\begin{align}
\epsilon _{TC}& =\Pr \left\{ \gamma _{\kappa }\geq \tau \right\} \bar{%
\varepsilon}+\Pr \left\{ \gamma _{\kappa }<\tau \right\} \int_{0}^{\infty
}xf_{\varepsilon _{\lambda }}\left( x\right) dx  \notag \\
& =\left[ 1-\left( 1-e^{-\frac{2\tau }{\bar{\gamma}}}\right) ^{N}\right]
\bar{\varepsilon}+\left( 1-e^{-\frac{2\tau }{\bar{\gamma}}}\right) ^{N}\bar{%
\varepsilon}~H_{N}.  \label{ETB}
\end{align}%
Solving (\ref{ETB}) with respect to $\tau $ yields%
\begin{equation}
\tau =\bar{\gamma}\ln \left[ \left( 1-\sqrt[N]{\frac{\epsilon _{TC}-\bar{%
\varepsilon}}{\bar{\varepsilon}\left( H_{N}-1\right) }}\right) ^{-1/2}\right]
.  \label{tau}
\end{equation}%
Substituting (\ref{tau}) into (\ref{CTB}), we can express $C_{TC}$ as a
function of $\epsilon _{TC}$, as shown below%
\begin{align}
& C_{TC}=\frac{\left( \frac{\epsilon _{TC}-\bar{\varepsilon}}{\bar{%
\varepsilon}\left( H_{N}-1\right) }\right) ^{\frac{N-1}{N}}}{2\ln \left(
2\right) }\Bigg[e^{\frac{2}{\bar{\gamma}}}\mathcal{E}_{1}\left( \frac{2}{%
\bar{\gamma}}\right) -e^{\frac{2}{\bar{\gamma}}}\mathcal{E}_{1}\left( \frac{2%
}{\bar{\gamma}}-\ln \left( 1-\sqrt[N]{\frac{\epsilon _{TC}-\bar{\varepsilon}%
}{\bar{\varepsilon}\left( H_{N}-1\right) }}\right) \right) \Bigg.  \notag \\
& \Bigg.-\left( 1-\sqrt[N]{\frac{\epsilon _{TC}-\bar{\varepsilon}}{\bar{%
\varepsilon}\left( H_{N}-1\right) }}\right) \ln \left( 1-\frac{\bar{\gamma}%
\ln \left( 1-\sqrt[N]{\frac{\epsilon _{TC}-\bar{\varepsilon}}{\bar{%
\varepsilon}\left( H_{N}-1\right) }}\right) }{2}\right) \Bigg]  \label{CTB02}
\\
& +\sum_{j=0}^{M-1}\frac{e^{\frac{2\left( j+1\right) \left[ 1-\frac{\bar{%
\gamma}}{2}\ln \left( 1-\sqrt[N]{\frac{\epsilon _{TC}-\bar{\varepsilon}}{%
\bar{\varepsilon}\left( H_{N}-1\right) }}\right) \right] }{\bar{\gamma}}}%
\mathcal{E}_{1}\left( \frac{2\left( j+1\right) \left[ 1-\frac{\bar{\gamma}}{2%
}\ln \left( 1-\sqrt[N]{\frac{\epsilon _{TC}-\bar{\varepsilon}}{\bar{%
\varepsilon}\left( H_{N}-1\right) }}\right) \right] }{\bar{\gamma}}\right)
+\ln \left[ 1-\frac{\bar{\gamma}}{2}\ln \left( 1-\sqrt[N]{\frac{\epsilon
_{TC}-\bar{\varepsilon}}{\bar{\varepsilon}\left( H_{N}-1\right) }}\right) %
\right] }{2\left[ N\binom{N-1}{j}\left[ 1-\sqrt[N]{\frac{\epsilon _{TC}-\bar{%
\varepsilon}}{\bar{\varepsilon}\left( H_{N}-1\right) }}\right] ^{j+1}\right]
^{-1}\left( -1\right) ^{j}\left( j+1\right) \ln \left( 2\right) .}  \notag
\end{align}

\begin{corollary}
\label{COR1}The tradeoff between ergodic capacity and average energy
transfer of the threshold-checking scheme is an increasing opportunity cost
tradeoff.
\end{corollary}

\begin{proof}
The proof follows by showing that $\partial ^{2}C_{TC}/\partial \epsilon
_{TC}^{2}<0$ for each $\bar{\varepsilon}<\epsilon _{TC}<H_{N}\bar{\varepsilon%
}$.
\end{proof}

\subsection{\label{WD0}Weighted Difference Selection Scheme}

Let us now focus on the case of $N=2$, i.e., the case where two relays are
available. For this case, the weighted difference selection method operates
as follows%
\begin{equation}
\gamma_{1}-\gamma_{2}%
\begin{array}{c}
\overset{s=1}{>} \\
\underset{s=2}{<}%
\end{array}
\nu\left( \varepsilon_{2}-\varepsilon_{1}\right)  \label{WD}
\end{equation}
where $\nu>0$ is a constant characterizing the tradeoff factor.

\subsubsection{Intuition Behind the Weighted Difference Scheme}

The intuition behind the weighted difference selection policy in (\ref{WD})
is simple: If one of the two relays experiences stronger links to $\mathcal{D%
}$ and $\mathcal{H}$ than the other, then this relay should be selected.
Otherwise, the relay with stronger channel to $\mathcal{D}$ is selected if
more priority is given to information transmission than to energy transfer,
and vice versa. This intuition is clearly illustrated in (\ref{WD}): If $%
\mathcal{R}_{1}$, for instance, is superior (inferior) to $\mathcal{R}_{2}$
in terms of both information and energy transfer, then the left hand side of
(\ref{WD}) will be positive (negative) and the right hand side of (\ref{WD})
negative (positive), leading to $s=1$ ($s=2$). If, on the other hand, $%
\mathcal{R}_{1}$ is superior (inferior) to $\mathcal{R}_{2}$ in terms of
information rate but inferior (superior) to $\mathcal{R}_{2}$ in terms of
energy transfer, then the two sides of (\ref{WD}) have the same sign hence
the selected relay is determined by the weighting coefficient, $\nu $.
Clearly, $\nu $ can take any positive real value therefore any desired
tradeoff factor can be achieved by properly adjusting $\nu $.

\subsubsection{Tradeoff Expression}

An expression for the tradeoff achieved by the weighted difference scheme is
provided in the ensuing Proposition.

\begin{proposition}
\label{Prop1}The capacity-energy tradeoff of the weighted difference scheme
is given by%
\begin{align}
C_{WD} & =\frac{2\left[ 1-\left( 1-\sqrt{\frac{\bar{\varepsilon}}{3\bar{%
\varepsilon}-2\epsilon_{WD}}}\right) ^{2}\right] \mathcal{E}_{1}\left( \frac{%
2}{\bar{\gamma}}\right) -e^{\frac{2}{\bar{\gamma}}}\mathcal{E}_{1}\left(
\frac{4}{\bar{\gamma}}\right) }{2e^{-\frac{2}{\bar{\gamma}}}\left[ 1-\left(
1-\sqrt{\frac{\bar{\varepsilon}}{3\bar{\varepsilon}-2\epsilon_{WD}}}\right)
^{2}\right] \ln\left( 2\right) }  \notag \\
& +\frac{\exp\left( -\frac{2}{\bar{\gamma}\left( 1-\sqrt{\frac {\bar{%
\varepsilon}}{3\bar{\varepsilon}-2\epsilon_{WD}}}\right) }\right) \left( 1-%
\sqrt{\frac{\bar{\varepsilon}}{3\bar{\varepsilon}-2\epsilon_{WD}}}\right)
^{2}\mathcal{E}_{1}\left( \frac{2\left( \frac{1}{1-\sqrt {\frac{\bar{%
\varepsilon}}{3\bar{\varepsilon}-2\epsilon_{WD}}}}-1\right) }{\bar{\gamma}}%
\right) }{2e^{-\frac{2}{\bar{\gamma}}}\left[ 1-\left( 1-\sqrt{\frac{\bar{%
\varepsilon}}{3\bar{\varepsilon}-2\epsilon_{WD}}}\right) ^{2}\right]
\ln\left( 2\right) }.  \label{WDTrad}
\end{align}
\end{proposition}

\begin{proof}
The proof is provided in Appendix \ref{PrProp}.
\end{proof}

\begin{corollary}
The tradeoff between ergodic capacity and average energy transfer of the
weighted difference scheme is an increasing opportunity cost tradeoff.
\end{corollary}

\begin{proof}
Similarly to Corollary \ref{COR1}, the proof follows by taking the second
derivative of $C_{WD}$ with respect to $\epsilon _{WD}$ and noting that $%
\partial ^{2}C_{WD}/\partial \epsilon _{WD}^{2}<0$ for each $\bar{\varepsilon%
}<\epsilon _{WD}<H_{N}\bar{\varepsilon}$.
\end{proof}

As shown later in Section \ref{NUM}, the weighted difference
scheme leads to a capacity-energy tradeoff that is close to the
Pareto frontier. However, it applies only for $N=2$. An extension
to $N>2$ is not straightforward since in that case multiple
comparisons among the candidate relays in terms of their
contributions to the overall capacity and the overall energy
transfer would be required.

\subsection{\label{IC0}On the CSI Requirements of the Proposed Schemes}

Here, we provide a brief discussion on the implementation complexity of the
three schemes under consideration, in terms of the amount of CSI required
for their operation. The time-sharing scheme has the lowest CSI requirement
of the three, because it requires CSI knowledge of either the $S$-$\mathcal{R%
}$-$\mathcal{D}$ or the $\mathcal{R}$-$\mathcal{H}$ links. The
threshold-checking scheme requires continuous CSI knowledge of the $S$-$%
\mathcal{R}$-$\mathcal{D}$ links; it additionally requires CSI knowledge of
the $\mathcal{R}$-$\mathcal{H}$ links for as long as the strength of the
end-to-end channel to $\mathcal{D}$ is below a given threshold. The weighted
difference scheme requires continuous CSI knowledge for all $S$-$\mathcal{R}$%
-$\mathcal{D}$ and $\mathcal{R}$-$\mathcal{H}$ links. Consequently, the
time-sharing scheme has relatively low, the threshold-checking scheme
medium, and the weighted difference scheme high requirements regarding CSI
knowledge. For the reader's convenience, the CSI requirements of the
proposed schemes are summarized in Table \ref{TABCSI}.

\begin{table}[ptb]
\caption{CSI Requirements and Application Scenarios of the Three
Considered Schemes } \label{TABCSI}\vspace{0.1cm}
\centering\renewcommand{\arraystretch}{1.6}
\begin{tabular}{|c|c|c|}
\hline \textbf{Relay Selection Scheme} & \textbf{CSI Requirements}
& \textbf{Number of Relays}
\\ \hline \emph{Time-Sharing
Scheme} & Low & $N\geq2$ \\ \hline \emph{Threshold-Checking
Scheme} & Medium & $N\geq2$ \\ \hline \emph{Weighted Difference
Scheme} & High & $N=2$ \\ \hline
\end{tabular}
\end{table}

\section{\label{OUT}Outage and Asymptotic Performance for a Given Energy
Transfer}

\subsection{Outage Probability}

For schemes with constant transmission rate where the source does not have
transmit-side CSI, an outage occurs if the end-to-end link to the
destination cannot support the transmission rate. In DF relaying, an outage
in the $\mathcal{S}$-$\mathcal{R}_{i}$-$\mathcal{D}$ link occurs if either
the $\mathcal{S}$-$\mathcal{R}_{i}$ or the $\mathcal{R}_{i}$-$\mathcal{D}$
link is in outage. Consequently, denoting the fixed transmission rate by $r$%
, an outage occurs if the end-to-end SNR of the active relay, $\gamma _{s}$,
drops below the threshold $\gamma _{\text{th}}=2^{2r}-1$. Here, we provide
expressions for the outage probability of the three considered schemes,
assuming that a prescribed amount of energy is transferred to $\mathcal{H}$.

\subsubsection{Outage Probability of the Time-Sharing Scheme}

The simplicity of the mode of operation of the time-sharing scheme shown in (%
\ref{TS1}) allows for a straightforward evaluation of the outage probability
as%
\begin{align}
P_{out,TS} & =\mu\Pr\left\{ \gamma_{\kappa}<\gamma_{\text{th}}\right\}
+\left( 1-\mu\right) \Pr\left\{ \gamma_{\lambda}<\gamma_{\text{th}}\right\}
\notag \\
& =\mu\left( 1-e^{-\frac{2\gamma_{\text{th}}}{\bar{\gamma}}}\right)
^{N}+\left( 1-\mu\right) \left( 1-e^{-\frac{2\gamma_{\text{th}}}{\bar{\gamma}%
}}\right) .  \label{OutTS1}
\end{align}
Substituting $\mu$ from (\ref{mu}) into (\ref{OutTS1}) yields the outage
probability of the time-sharing scheme as a function of $\epsilon_{TS}$ and $%
\bar{\varepsilon}$,%
\begin{equation}
P_{out,TS}=\frac{e^{-\frac{2\gamma_{\text{th}}}{\bar{\gamma}}}\left\{ 1-%
\frac {\epsilon_{TS}}{\bar{\varepsilon}}+e^{\frac{2\gamma_{\text{th}}}{\bar{%
\gamma}}}\left[ \frac{\epsilon_{TS}}{\bar{\varepsilon}}+\left( H_{N}-\frac {%
\epsilon_{TS}}{\bar{\varepsilon}}\right) \left( 1-e^{-\frac{2\gamma_{\text{th%
}}}{\bar{\gamma}}}\right) ^{N}-1\right] \right\} }{H_{N}-1}.  \label{OutTS2}
\end{equation}
Alternatively, using (\ref{delta}) we can express the outage probability as
a function of the tradeoff factor, $\delta$, as%
\begin{equation}
P_{out,TS}=\left( 1-\delta\right) \left( 1-e^{-\frac{2\gamma_{\text{th}}}{%
\bar{\gamma}}}\right) ^{N}+\delta\left( 1-e^{-\frac{2\gamma_{\text{th}}}{%
\bar{\gamma}}}\right) .  \label{OutTS3}
\end{equation}
In fact, the expression in (\ref{OutTS3}) was expected since it is clear
from (\ref{TS1}) that for the time-sharing scheme, $\delta=1-\mu$.

\subsubsection{Outage Probability of the Threshold-Checking Scheme}

Considering the relay selection policy in the threshold-checking scheme in (%
\ref{TB1}), for calculating the outage probability we consider the following
cases:

\begin{itemize}
\item If $\tau\leq\gamma_{\text{th}}$, then the outage performance is
determined by $\gamma_{\kappa}$, hence an outage occurs if $%
\gamma_{\kappa}<\gamma_{\text{th}}$.

\item If $\tau>\gamma_{\text{th}}$, then an outage occurs if $%
\gamma_{\lambda}<\gamma_{\text{th}}$. The outage event in this case is thus
equivalent to the event where the SNR of a randomly selected relay -- out of
the pool of $N$ relays -- is smaller than $\gamma_{\text{th}}$, while the
SNRs of the remaining $N-1$ relays are all smaller than $\tau$.
\end{itemize}

The overall outage probability is thus expressed as%
\begin{equation}
P_{out,TC}=\left\{
\begin{array}{c}
\left( 1-e^{-\frac{2\gamma_{\text{th}}}{\bar{\gamma}}}\right) ^{N}\text{, \
\ if }\tau\leq\gamma_{\text{th}} \\
\left( 1-e^{-\frac{2\gamma_{\text{th}}}{\bar{\gamma}}}\right) \left( 1-e^{-%
\frac{2\tau}{\bar{\gamma}}}\right) ^{N-1}\text{, \ \ if }\tau >\gamma_{\text{%
th}}%
\end{array}
.\right.  \label{OutTB1}
\end{equation}
By replacing $\tau$ in (\ref{OutTB1}) with the right-hand side of (\ref{tau}%
) we can express the outage probability of the threshold-checking scheme as
a function of $\epsilon_{TC}$ and $\bar{\varepsilon}$,%
\begin{equation}
P_{out,TC}=\left\{
\begin{array}{c}
\left( 1-e^{-\frac{2\gamma_{\text{th}}}{\bar{\gamma}}}\right) ^{N}\text{, \
\ if }\frac{\epsilon_{TC}}{\bar{\varepsilon}}\leq1+\left( H_{N}-1\right)
\left( 1-e^{-\frac{2\gamma_{\text{th}}}{\bar{\gamma}}}\right) ^{N} \\
\left( 1-e^{-\frac{2\gamma_{\text{th}}}{\bar{\gamma}}}\right) \left( \frac {%
\frac{\epsilon_{TC}}{\bar{\varepsilon}}-1}{H_{N}-1}\right) ^{\frac{N-1}{N}}%
\text{, \ \ if }\frac{\epsilon_{TC}}{\bar{\varepsilon}}>1+\left(
H_{N}-1\right) \left( 1-e^{-\frac{2\gamma_{\text{th}}}{\bar{\gamma}}}\right)
^{N}%
\end{array}
.\right.  \label{OutTB0}
\end{equation}
Alternatively, $P_{out,TC}$ is expressed as a function of the tradeoff
factor as%
\begin{equation}
P_{out,TC}=\left\{
\begin{array}{c}
\left( 1-e^{-\frac{2\gamma_{\text{th}}}{\bar{\gamma}}}\right) ^{N}\text{, \
\ if }\delta\leq\left( 1-e^{-\frac{2\gamma_{\text{th}}}{\bar{\gamma}}%
}\right) ^{N} \\
\left( 1-e^{-\frac{2\gamma_{\text{th}}}{\bar{\gamma}}}\right) \delta^{\frac{%
N-1}{N}}\text{, \ \ if }\delta>\left( 1-e^{-\frac{2\gamma_{\text{th}}}{\bar{%
\gamma}}}\right) ^{N}%
\end{array}
.\right.  \label{OutTB2}
\end{equation}

\subsubsection{Outage Probability of the Weighted Difference Scheme}

\begin{proposition}
\label{Theo2}The outage probability of the weighted difference scheme is
given by
\begin{equation}
P_{out,WD}=\frac{e^{-\frac{4\gamma_{\text{th}}}{\bar{\gamma}}}\left( e^{%
\frac {2\gamma_{\text{th}}}{\bar{\gamma}}}-1\right) ^{2}+\left( 1-\sqrt{%
\frac {\bar{\varepsilon}}{3\bar{\varepsilon}-2\epsilon_{WD}}}\right) ^{2}%
\left[ e^{\frac{2\gamma_{\text{th}}}{\bar{\gamma}}}\left( 2-e^{\frac{%
2\gamma_{\text{th}}}{\bar{\gamma}\left( 1-\sqrt{\frac{\bar{\varepsilon}}{3%
\bar{\varepsilon }-2\epsilon_{WD}}}\right) }}\right) -1\right] }{1-\left( 1-%
\sqrt {\frac{\bar{\varepsilon}}{3\bar{\varepsilon}-2\epsilon_{WD}}}\right)
^{2}}  \label{OutWD0}
\end{equation}
or, in terms of the tradeoff factor, as%
\begin{equation}
P_{out,WD}=\frac{e^{-\frac{4\gamma_{\text{th}}}{\bar{\gamma}}}\left( e^{%
\frac {2\gamma_{\text{th}}}{\bar{\gamma}}}-1\right) ^{2}+\left( 1-\sqrt{%
\frac{1}{1-\delta}}\right) ^{2}\left[ e^{\frac{2\gamma_{\text{th}}}{\bar{%
\gamma}}}\left( 2-e^{\frac{2\gamma_{\text{th}}}{\bar{\gamma}\left( 1-\sqrt{%
\frac{1}{1-\delta}}\right) }}\right) -1\right] }{1-\left( 1-\sqrt{\frac{1}{%
1-\delta}}\right) ^{2}}.  \label{OutWD02}
\end{equation}
\end{proposition}

\begin{proof}
The proof is provided in Appendix \ref{AppOut}.
\end{proof}


\subsection{Asymptotic Analysis}

By taking the Taylor series expansion and keeping only the first order
terms, (\ref{OutTS3}), (\ref{OutTB2}), and (\ref{OutWD02}) reduce after some
algebraic manipulations to the following high-SNR expressions%
\begin{eqnarray}
P_{out,TS} &\approx &\frac{2\gamma _{\text{th}}}{\bar{\gamma}}\delta
\label{ASTS} \\
P_{out,TC} &\approx &\frac{2\gamma _{\text{th}}}{\bar{\gamma}}\delta ^{\frac{%
N-1}{N}}  \label{ASTB} \\
P_{out,WD} &\approx &\frac{2\gamma _{\text{th}}}{\bar{\gamma}}\left( 1-\sqrt{%
1-\delta }\right) .  \label{ASWD}
\end{eqnarray}%
Using the fact that $1-\sqrt{1-\delta }<\delta <\sqrt{\delta }$ for $%
0<\delta <1$, we observe from (\ref{ASTS})-(\ref{ASWD}) that for $N=2$ and
for any tradeoff factor the best asymptotic outage performance is achieved
by the weighted difference scheme, the time-sharing scheme performs in the
middle of the other two, and the worst asymptotic outage performance is
achieved by the threshold-checking scheme. The same result holds also for $%
N>2$, as $\delta <\delta ^{\frac{N-1}{N}}$ for any $0<\delta <1$. As will be
shown via numerical examples in Section \ref{NUM}, the outcome of this
comparison is different from that in terms of the ergodic capacity, as for
the latter comparison the threshold-checking scheme outperforms the
time-sharing scheme.

Using (\ref{ASTS})-(\ref{ASWD}), the diversity gain, $\mathcal{G}_{d}$, and
array gain, $\mathcal{G}_{a}$, can be straightforwardly derived by
expressing the asymptotic outage probability in the form $P_{out}=\left(
\mathcal{G}_{a}\bar{\gamma}/\gamma_{\text{th}}\right) ^{-\mathcal{G}_{d}}$
\cite{J:Giannakis_Param_Perf}. The results are summarized in the ensuing two
Corollaries.

\begin{corollary}
\label{AS1}The diversity order of the time-sharing, threshold-checking, and
weighted difference schemes equals one, unless a zero tradeoff factor is
employed. In other words, if the required energy transfer to $\mathcal{H}$
pertaining to the above mentioned selection schemes is larger (even by an
infinitesimally small amount) than its lower boundary, $\epsilon_{\min}$,
then the diversity gain is lost.
\end{corollary}

\begin{corollary}
\label{AS2}The array gain of the time-sharing, threshold-checking and
weighted difference scheme equal respectively%
\begin{eqnarray}
\mathcal{G}_{a,TS} &=&\frac{1}{2\delta }  \label{AGTS} \\
\mathcal{G}_{a,TC} &=&\frac{1}{2\delta ^{\frac{N-1}{N}}}  \label{AGTB} \\
\mathcal{G}_{a,WD} &=&\frac{1}{2\left( 1-\sqrt{1-\delta }\right) }.
\label{AGWD}
\end{eqnarray}
\end{corollary}

\section{\label{PAR}Pareto Efficiency for $N=2$}

The idea behind the weighted difference scheme is general enough so that by
a careful amendment of the quantities in the left hand side of (\ref{WD}),
we can choose to optimize any long-term performance metric associated with
information transmission to $\mathcal{D}$, for a given energy transfer to $%
\mathcal{H}$ and $N=2$. This interesting conclusion is summarized in the
ensuing Theorem.

\begin{theorem}
\label{Theo}Let $\mathcal{F}$ denote any metric that characterizes the
long-term performance of the information transmission to $\mathcal{D}$, in
the sense that the performance is optimized when $\mathcal{F}$ is maximized.
Let $\mathcal{F}\left( \gamma_{i}\right) $ be a non-decreasing function of $%
\gamma_{i}$, which describes the instantaneous realization of $\mathcal{F}$
associated with the use of the $\mathcal{S}$-$\mathcal{R}_{i}$-$\mathcal{D}$
link. The Pareto frontier of the tradeoff between $\mathcal{F}$ and the
average energy transfer is achieved by the following selection policy%
\begin{equation}
\mathcal{F}\left( \gamma_{1}\right) -\mathcal{F}\left( \gamma_{2}\right)
\begin{array}{c}
\overset{s=1}{>} \\
\underset{s=2}{<}%
\end{array}
\zeta\left( \varepsilon_{2}-\varepsilon_{1}\right)  \label{PO}
\end{equation}
where $\zeta>0$ is a constant, in which the tradeoff factor is reflected.
\end{theorem}

\begin{proof}
The proof is provided in Appendix \ref{ProofTheo}.
\end{proof}

The utility parameter $\mathcal{F}$ in (\ref{PO}) can represent any of the
most common performance metrics whose maximization is associated with
optimizing system performance, such as the average SNR, the ergodic
capacity, the probability of no-outage, and the probability of correct bit
detection. If the average SNR is the metric of interest, then the weighted
difference scheme is Pareto efficient because (\ref{PO}) reduces to (\ref{WD}%
). Next, we investigate the tradeoffs pertaining to the optimal
ergodic capacity and the optimal probability of no-outage for a
given energy transfer. The analysis for the optimal probability of
correct bit detection follows similarly, and is omitted here for
brevity.

\subsection{Optimal Ergodic Capacity for a Given Energy Transfer}

Theorem \ref{Theo} provides the Pareto frontier of the tradeoff between
ergodic capacity and average transferred energy by substituting $\mathcal{F}%
\left( \gamma_{i}\right) $ with the instantaneous capacity expression, i.e.,
by setting
\begin{equation}
\mathcal{F}\left( \gamma_{1}\right) =\frac{1}{2}\log_{2}\left( 1+\gamma
_{1}\right) \text{, \ }\mathcal{F}\left( \gamma_{2}\right) =\frac{1}{2}%
\log_{2}\left( 1+\gamma_{2}\right)
\end{equation}
in (\ref{PO}).

\subsubsection*{\label{PARC}The Pareto Frontier of Ergodic Capacity Vs.
Energy Transfer}

Due to the complicated mathematical analysis involved, the exact derivation
of the Pareto frontier of the desired tradeoff is cumbersome. In fact, by
following a similar approach as that in Appendix \ref{PrProp} (the
parameters $\left( \gamma _{2}-\gamma _{1}\right) /\nu $ and $\left( \gamma
_{1}-\gamma _{2}\right) /\nu $ at the integral limits of the fourth integral
in $\mathcal{I}_{2a}$, $\mathcal{I}_{2b}$, $\mathcal{I}_{3a}$, $\mathcal{I}%
_{3b}$ in (\ref{eWD}) are substituted by $\frac{1}{2\zeta }\log _{2}\left(
1+\gamma _{2}\right) -\frac{1}{2\zeta }\log _{2}\left( 1+\gamma _{1}\right) $
and $\frac{1}{2\zeta }\log _{2}\left( 1+\gamma _{1}\right) -\frac{1}{2\zeta }%
\log _{2}\left( 1+\gamma _{2}\right) $, respectively), we can express $%
\epsilon $ as a function of $\zeta $, yet that expression involves a
two-fold integration which, to the best of our knowledge, is solvable only
by numerical methods. Thus, the resulting Pareto frontier curve can be
derived numerically only. Further discussions on this Pareto frontier are
provided in Section \ref{NUM}.

\subsection{Optimal Outage Probability for a Given Energy Transfer}

In case the outage probability is the metric of interest, then $\mathcal{F}%
\left( \gamma _{i}\right) $ in (\ref{PO}) is substituted by
\begin{equation}
\mathcal{F}\left( \gamma _{1}\right) =\left\{
\begin{array}{c}
1\text{, if }\gamma _{1}>\gamma _{\text{th}} \\
0\text{, if }\gamma _{1}<\gamma _{\text{th}}%
\end{array}%
\right. \text{, }\mathcal{F}\left( \gamma _{2}\right) =\left\{
\begin{array}{c}
1\text{, if }\gamma _{2}>\gamma _{\text{th}} \\
0\text{, if }\gamma _{2}<\gamma _{\text{th}}%
\end{array}%
\right. .  \label{POOut}
\end{equation}%
That is, for optimizing the outage probability for a given energy transfer
the function $\mathcal{F}\left( \gamma _{i}\right) $ in (\ref{PO}) reduces
to the binary event of no-outage, given the instantaneous realization of $%
\gamma _{i}$, $i=1,2$. The resulting tradeoff between the probability of
no-outage and the average energy transfer is Pareto efficient, and is
investigated below.

\subsubsection*{The Pareto Frontier of Probability of No-Outage Vs. Energy
Transfer}

By following a similar analysis as that in Appendix \ref{AppOut}, the
average transferred energy, $\epsilon _{out,opt}$, can be expressed as%
\begin{align}
\epsilon _{out,opt}& =\underset{\text{Case of }\gamma _{1}>\gamma _{th}\text{%
;}\gamma _{2}>\gamma _{\text{th}}}{~~\underbrace{\frac{3}{2}e^{-\frac{%
4\gamma _{\text{th}}}{\bar{\gamma}}}\bar{\varepsilon}}}+~2~\underset{\text{%
Cases of }\gamma _{1}<\gamma _{\text{th}}\text{;}\gamma _{2}>\gamma _{\text{%
th}}\text{ , }\gamma _{1}>\gamma _{\text{th}}\text{;}\gamma _{2}<\gamma _{%
\text{th}}}{\underbrace{\frac{e^{-\frac{4\gamma _{\text{th}}}{\bar{\gamma}}-%
\frac{1}{\zeta \bar{\varepsilon}}}\left( e^{\frac{2\gamma _{\text{th}}}{\bar{%
\gamma}}}-1\right) \left[ \zeta \bar{\varepsilon}\left( 2e^{\frac{1}{\zeta
\bar{\varepsilon}}}+1\right) +1\right] }{2\zeta }}}+\underset{\text{Case of }%
\gamma _{1}<\gamma _{\text{th}}\text{;}\gamma _{2}<\gamma _{\text{th}}}{%
\underbrace{\frac{3}{2}\left( 1-e^{-\frac{4\gamma _{\text{th}}}{\bar{\gamma}}%
}\right) ^{2}\bar{\varepsilon}}}  \notag \\
& =\frac{e^{-\frac{4\gamma _{\text{th}}}{\bar{\gamma}}}}{2}\left[ \bar{%
\varepsilon}\left( 2-2e^{\frac{2\gamma _{\text{th}}}{\bar{\gamma}}}+3e^{%
\frac{4\gamma _{\text{th}}}{\bar{\gamma}}}\right) +\frac{2e^{-\frac{1}{\zeta
\bar{\varepsilon}}}\left( \zeta \bar{\varepsilon}+1\right) \left( e^{\frac{%
2\gamma _{\text{th}}}{\bar{\gamma}}}-1\right) }{\zeta }\right] .
\label{eoutopt}
\end{align}%
An important observation that is made from (\ref{eoutopt}) is that the lower
boundary of the average transferred energy in this case is different from
the lower boundary shown in Lemma \ref{lem2}. In particular, by taking
limits in (\ref{eoutopt}) for $\zeta \rightarrow 0^{+}$, we obtain\footnote{%
In fact, $\epsilon _{out,opt,\min }$ experiences a discontinuity for $\zeta
=0$, as for this case it is clear from (\ref{PO}) that the scheme reduces to
ignoring the energy transfer when selecting the relay, resulting in $%
\epsilon _{out,opt,\min }=\bar{\varepsilon}$. This case, however, is
excluded from our analysis as it does not result in any tradeoff between
probability of no-outage and energy transfer.}%
\begin{equation}
\epsilon _{out,opt,\min }=\bar{\varepsilon}\left( \frac{3}{2}+e^{-\frac{%
4\gamma _{\text{th}}}{\bar{\gamma}}}-e^{-\frac{2\gamma _{\text{th}}}{\bar{%
\gamma}}}\right) .  \label{eoptmin}
\end{equation}%
This result reveals that the selection policy in (\ref{PO}) (in conjunction
with (\ref{POOut})) achieves the Pareto frontier for a limited range of
tradeoff factor. In particular, it follows from (\ref{delta}) and (\ref%
{eoptmin}) that the tradeoff factor in this case spans the interval%
\begin{equation}
\delta \in \left[ 1-2\left( e^{-\frac{2\gamma _{\text{th}}}{\bar{\gamma}}%
}-e^{-\frac{4\gamma _{\text{th}}}{\bar{\gamma}}}\right) ,1\right] .
\label{deltarange}
\end{equation}%
The limited range of the tradeoff factor can be explained by the fact that a
tradeoff in this case exists only for $\mathcal{F}\left( \gamma _{1}\right)
\not=\mathcal{F}\left( \gamma _{2}\right) $, since for the complementary
event of $\mathcal{F}\left( \gamma _{1}\right) =\mathcal{F}\left( \gamma
_{2}\right) $ the left-hand side of (\ref{PO}) equals zero, hence the relay
selection policy is independent of $\zeta $. In other words, no exchange
takes place between outage probability and energy transfer as long as $%
\mathcal{F}\left( \gamma _{1}\right) =\mathcal{F}\left( \gamma _{2}\right) $%
, an event which occurs with probability $1-2\left( e^{-\frac{2\gamma _{%
\text{th}}}{\bar{\gamma}}}-e^{-\frac{4\gamma _{\text{th}}}{\bar{\gamma}}%
}\right) $. The practical meaning of this observation is that the selection
policy in (\ref{PO}), (\ref{POOut}), achieves an energy transfer increase
from $\bar{\varepsilon}$ to $\epsilon _{out,opt,\min }$ with no outage cost.
This fact will be better explained through numerical examples in Section \ref%
{NUMOut}.

Since (\ref{eoutopt}) is not solvable with respect to $\zeta $, a
mathematical expression for the Pareto probability of no-outage as a
function of the average energy transfer is not possible. Hence, we confine
ourselves to obtaining an expression for the probability of no-outage as a
function of $\zeta $, as follows. The no-outage event occurs for the
following cases: $\gamma _{1}>\gamma _{\text{th}}$ and $\gamma _{2}>\gamma _{%
\text{th}}$; $\gamma _{1}>\gamma _{\text{th}}$ and $\gamma _{2}<\gamma _{%
\text{th}}$ and $\varepsilon _{2}<\varepsilon _{1}+1/\zeta $; $\gamma
_{1}<\gamma _{\text{th}}$ and $\gamma _{2}>\gamma _{\text{th}}$ and $%
\varepsilon _{1}<\varepsilon _{2}+1/\zeta $. The probability of the union of
these events can be evaluated by solving the corresponding integrals, which
are similar to the integrals in (\ref{eWD}), yielding%
\begin{equation}
P_{no-out,opt}=e^{-\frac{4\gamma _{\text{th}}}{\bar{\gamma}}}\left[ 2e^{%
\frac{2\gamma _{\text{th}}}{\bar{\gamma}}}-e^{-\frac{1}{\zeta \bar{%
\varepsilon}}}\left( e^{\frac{2\gamma _{\text{th}}}{\bar{\gamma}}}-1\right)
-1\right] .  \label{No_Out1}
\end{equation}

\section{\label{NUM}Numerical Examples}

This section presents a set of illustrative examples that provide insight
into the behavior of the ergodic capacity and the outage probability, for a
given required average energy transfer to $\mathcal{H}$.

\subsection{Ergodic Capacity vs. Energy Transfer}

Fig. \ref{Fig_Trad1} depicts the ergodic capacity vs. the average
energy
transferred to $\mathcal{H}$ (normalized with respect to $\bar{\varepsilon}$%
), for the case of two available relays $N=2$ and $\bar{\gamma}=20$ dB. The
performance of the three tradeoff schemes studied in Section \ref{TRAD} is
compared with the Pareto frontier developed in Section \ref{PAR}, for the
range of feasible transmitted energy values, i.e., for $\bar{\varepsilon }%
<\epsilon<H_{2}\bar{\varepsilon}=1.5\bar{\varepsilon}$. The curves
pertaining to the three schemes in Section \ref{TRAD} were obtained from (%
\ref{CTS2}), (\ref{CTB02}), and (\ref{WDTrad}). For the generation of the
Pareto frontier curve numerical methods were used for obtaining the value of
$\zeta$ that leads to certain average energy transfer, as described in
Section \ref{PARC}; this value of $\zeta$ was then used for obtaining
numerical values for the corresponding ergodic capacity. As expected, the
time-sharing scheme leads to a linear tradeoff curve (equal opportunity
cost) which lies below the tradeoff curves of the threshold-checking and
weighted difference scheme, since the two latter schemes achieve tradeoffs
with increasing opportunity costs. Moreover, it is observed that the
weighted difference scheme approaches the Pareto frontier in the region
close to the tradeoff boundaries, and that it clearly outperforms the
time-sharing and the threshold-checking schemes.

\begin{figure}[ptb]
\centering
\includegraphics[trim=.4cm .5cm .5cm .4cm, clip=true,
scale=1.3]{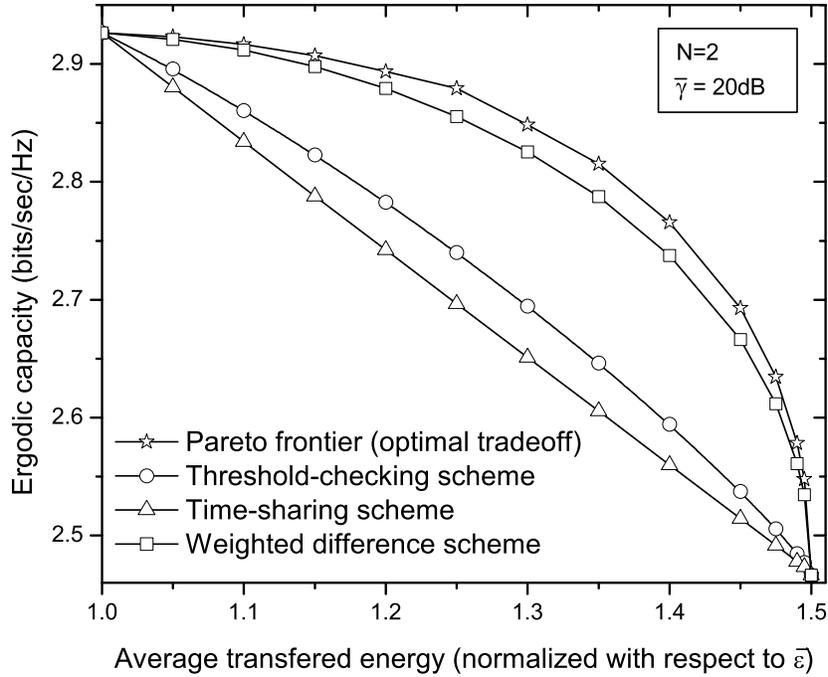} \caption{The capacity-energy
tradeoff of the schemes under consideration, for
$\bar{\protect\gamma}=20$ dB.} \label{Fig_Trad1}
\end{figure}

Similar observations are obtained from Fig. \ref{Fig_Trad2}, where
the same
tradeoffs are now plotted vs. the tradeoff factor, $\delta $, for $N=2$ and $%
\bar{\gamma}=10$ dB. We note that the weighted difference curve approximates
the Pareto frontier for low $\bar{\gamma}$, for the entire range of $\delta $%
. Moreover, the steepness of the curve in the region close to the boundaries
reveals that a large gain in ergodic capacity (transferred energy) is
attained without much sacrifice in energy transfer (ergodic capacity), for $%
\delta $ approaching zero or one. Furthermore, Fig. \ref{Fig_Trad2} shows
that the theoretically derived tradeoff results pertaining to the schemes
considered in Section \ref{TRAD} are in agreement with simulations.

\begin{figure}[ptb]
\centering
\includegraphics[trim=.4cm .5cm .5cm .4cm, clip=true,
scale=1.3]{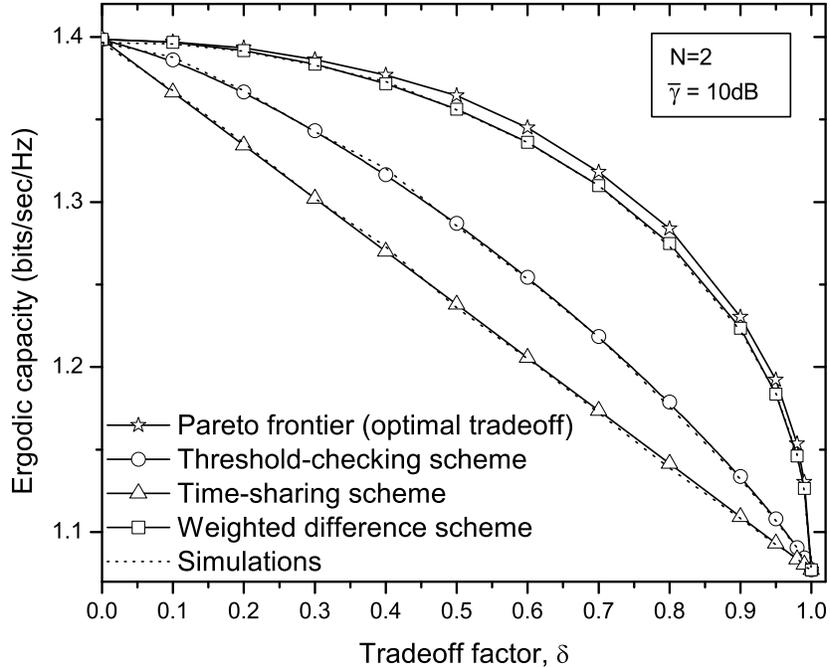} \caption{The capacity-energy
tradeoff of the schemes under consideration plotted vs. the
tradeoff factor, for $\bar{\protect\gamma }=10$ dB.}
\label{Fig_Trad2}
\end{figure}

The ergodic capacity vs. the average SNR per link, $\bar{\gamma}$,
for several values of $\delta $, is illustrated in Fig.
\ref{Fig_CapG}. We observe that by increasing the tradeoff factor
from $\delta =0$ to $\delta =1 $, a capacity decrease occurs,
which corresponds to an SNR loss of approximately $3$ dB. That is,
the cost in terms of capacity for increasing the wireless energy
transfer to $\mathcal{H}$ from its minimum to its maximum possible
value is approximately $3$ dB. This capacity cost is reduced if a
tradeoff factor smaller than one is selected. Moreover, we note
that the capacity of the weighted difference scheme approximates
that of the Pareto efficient scheme in (\ref{PO}); particularly
for low SNRs, the weighted difference scheme is almost Pareto
efficient.

\subsection{\label{NUMOut}Outage and No-Outage Probability vs. Energy
Transfer}

Fig. \ref{Fig_TradOut} illustrates the tradeoff between the probability of
no-outage and the average energy transfer, for $N=2$ and $\bar{\gamma}%
=2\gamma_{\text{th}}/\ln\left( 2\right) $. This particular choice for $\bar{%
\gamma }$ was made for convenience of presentation, since it follows from (%
\ref{delta}) that this choice of $\bar{\gamma}$ maximizes the range of the
feasible tradeoff factor for the Pareto efficient scheme, yielding $\delta
\in\left[ 0.5,1\right] $. The main observations drawn from Fig. \ref%
{Fig_TradOut} are the following: a) The weighted difference scheme
outperforms the threshold-checking and the time-sharing scheme, except for
small values of $\delta$. b) The threshold-checking scheme achieves Pareto
efficiency for small $\delta$, yet its performance is degraded for large $%
\delta$, where it approaches the performance of the time-sharing scheme. c)
The Pareto efficient scheme in (\ref{PO}) achieves the maximum feasible
probability of no-outage at its lower boundary (i.e., for $\delta=0.5$ for
the case of $\bar{\gamma}=2\gamma_{\text{th}}/\ln\left( 2\right) $). The
interpretation of this observation is as follows. It is trivial to prove
that by using $\zeta=0$ in (\ref{PO}) the optimum outage performance is
achieved, since the relay selection is based solely on the ability to
achieve an overall SNR larger than $\gamma_{\text{th}}$; the energy transfer
equals $\epsilon=\bar {\varepsilon}$. However, observation c) reveals that
by increasing $\zeta$ in (\ref{PO}) by an infinitesimally small amount, we
can increase the transferred energy from $\bar{\varepsilon}$ to $\bar{%
\varepsilon}\left( \frac{3}{2}+e^{-\frac{4\gamma_{\text{th}}}{\bar{\gamma}}%
}-e^{-\frac{2\gamma_{\text{th}}}{\bar {\gamma}}}\right) $, as (\ref{eoptmin}%
) suggests. In other words, we can offer more energy transfer to $\mathcal{H}
$ with no outage cost. In fact, this phenomenon stems from the on-off nature
of the outage events.

\begin{figure}[ptb]
\centering
\includegraphics[trim=.4cm .5cm .5cm 1.cm, clip=true,
scale=1.3]{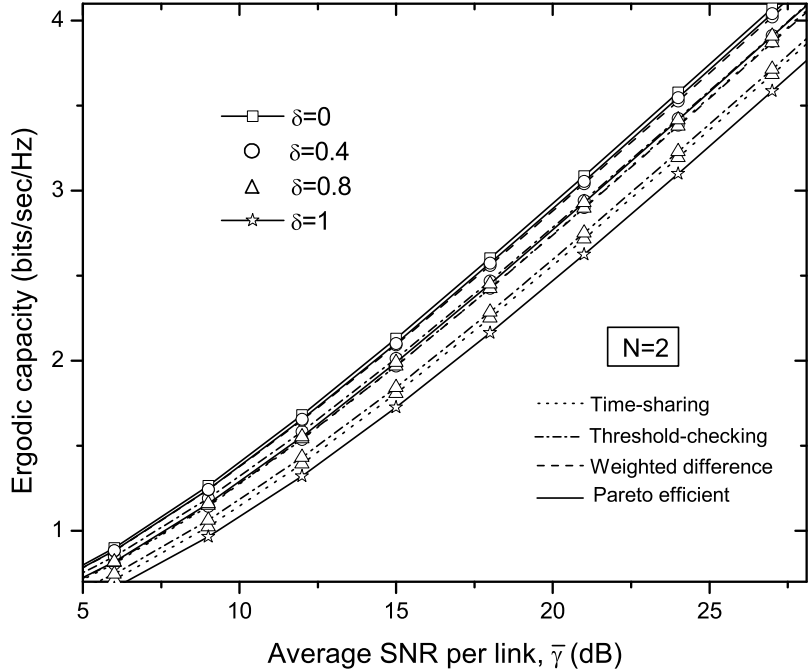} \caption{Ergodic capacity vs.
average SNR per link, $\bar{\protect\gamma}$, for $N=2$ and
several values of the tradeoff factor, $\protect\delta$.}
\label{Fig_CapG}
\end{figure}

The outage probability of the schemes under consideration are
depicted in Fig. \ref{Fig_OutG}, for $N=2$ and some values of
$\delta$ (for the Pareto
efficient scheme, if $\delta$ is smaller than its lower bound in (\ref{delta}%
) then this lower bound was used since the energy transfer to $\mathcal{H}$
satisfies the required energy transfer that $\delta$ implies). The main
conclusion drawn from Fig. \ref{Fig_OutG} is that, as suggested by Corollary %
\ref{AS1}, the slope of the outage curves of the time-sharing, the
threshold-checking, and the weighted difference schemes is negative unity
(in a log-log scale) for $\delta\not =0$, implying unit diversity order.
However, the slope of the Pareto efficient scheme demonstrates full
diversity order of $\mathcal{G}_{d}=2$, unless the maximum tradeoff factor
is allocated ($\delta=1$). This result was also expected since for $\delta=1$
the relay selection is made based solely on the strength of the $\mathcal{R}%
_{i}$-$\mathcal{H}$ channels. Moreover, we observe that the performance of
the threshold-checking scheme is inferior to all its counterparts in the
medium and high $\bar{\gamma}/\gamma_{\text{th}}$ region (equivalently, the
medium and low outage probability region), a fact which corroborates
Corollary \ref{AS2}. This behavior is in contrast to that in the low $\bar{%
\gamma}/\gamma_{\text{th}}$ region (i.e., for $\bar{\gamma}/\gamma_{\text{th}%
}<5$ dB), where the threshold-checking scheme outperforms the time-sharing
and the weighted-difference scheme, and actually approaches the behavior of
the Pareto efficient scheme.

\begin{figure}[ptb]
\centering
\includegraphics[trim=.4cm .5cm .5cm .4cm, clip=true,
scale=1.3]{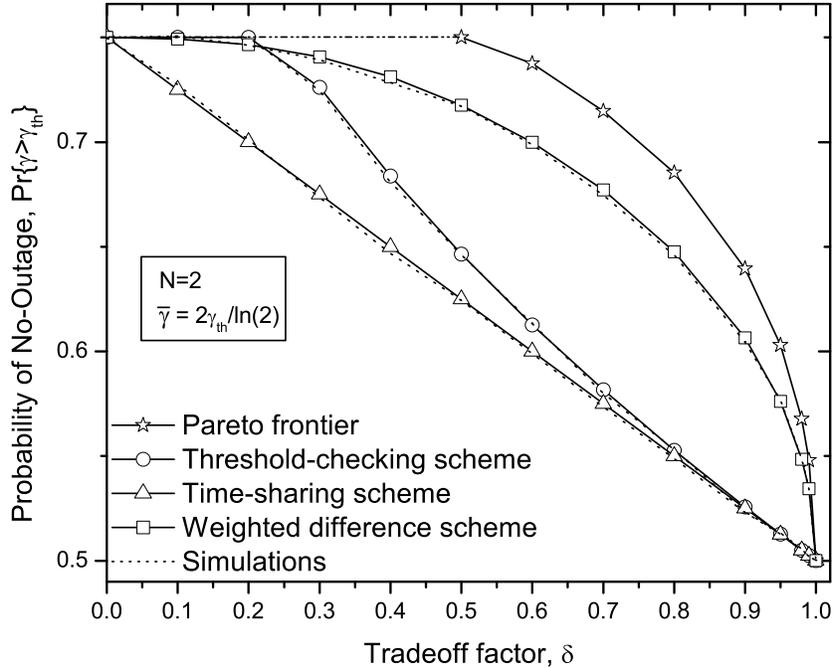} \caption{The tradeoff between the
probability of no-outage and the average
energy transfer vs. the tradeoff factor, for $N=2$ and $\bar{\protect\gamma}%
=2\protect\gamma_{\text{th}}/\text{ln}(2)$.} \label{Fig_TradOut}
\end{figure}

Finally, Fig. \ref{Fig_OutGN3} deals with the case of three
participating relays ($N=3$), and shows the outage behavior of the
time-sharing and the threshold-checking scheme. We observe that
the time-sharing scheme
outperforms the threshold-checking scheme for high values of $\bar{\gamma}%
/\gamma _{\text{th}}$, yet the threshold-checking scheme performs slightly
better for low values of $\bar{\gamma}/\gamma _{\text{th}}$. Moreover, we
notice that even the slightest increase of the required energy transfer to $%
\mathcal{H}$ (as this is reflected by setting $\delta =0.01$) severely
deteriorates the asymptotic outage performance, corroborating thus Corollary %
\ref{AS1}. Nevertheless, the shift of the outage curves towards the negative
unit slope occurs for relatively high SNRs for $\delta \rightarrow 0$,
implying that the outage curves maintain their diversity characteristics in
the medium SNR region when $\delta $ approaches zero.

\begin{figure}[ptb]
\centering
\includegraphics[trim=.4cm .5cm .5cm .4cm, clip=true,
scale=1.3]{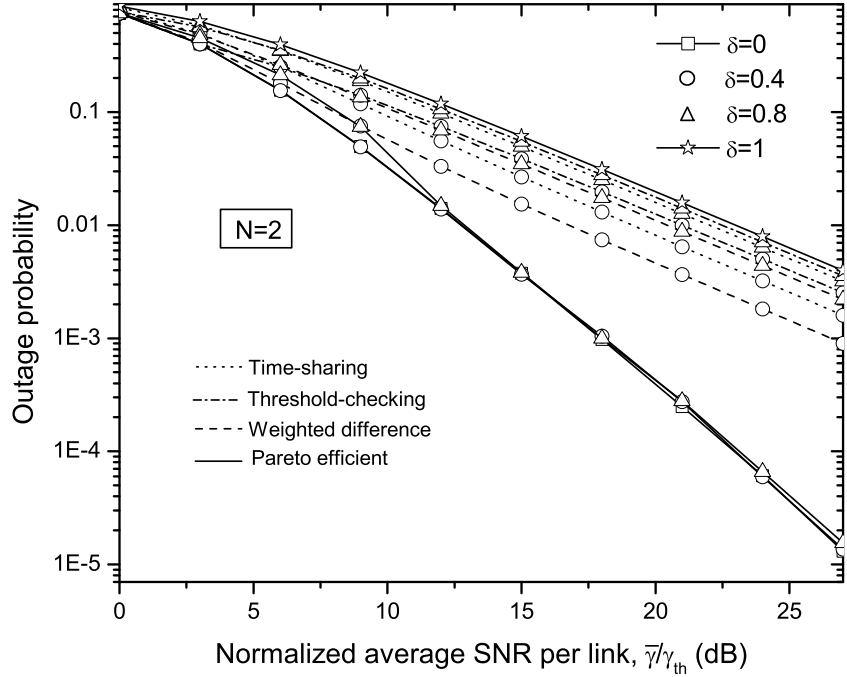}
\caption{Outage probability vs. normalized average SNR per link, $\bar{%
\protect\gamma}/\protect\gamma_{\text{th}}$, for $N=2$ and several
values of the tradeoff factor, $\protect\delta$.} \label{Fig_OutG}
\end{figure}

\section{Conclusions}

In scenarios involving relay-assisted information and energy transfer to a
designated receiver and a designated RF energy harvester, respectively, the
policy regarding the activated relay determines the tradeoff between quality
of information transfer and wireless energy transfer. We provided a thorough
analysis of this tradeoff for i.i.d. Rayleigh fading channels. For the
versatile scenario of $N$ candidate relays, \textquotedblleft time-sharing
selection\textquotedblright\ and \textquotedblleft threshold-checking
selection\textquotedblright\ schemes were developed and analyzed. Numerical
results showed that \textquotedblleft threshold-checking
selection\textquotedblright\ is better in terms of achieved capacity for a
given required energy transfer. However, in terms of outage probability for
a given energy transfer, \textquotedblleft time-sharing
selection\textquotedblright\ outperforms \textquotedblleft
threshold-checking selection\textquotedblright\ when the normalized average
SNR per link (with respect to the outage threshold SNR) is greater than $5$
dB; for low SNRs, the outcome of the comparison is reversed.

\begin{figure}[ptb]
\centering
\includegraphics[trim=.4cm .5cm .5cm .4cm, clip=true,
scale=1.3]{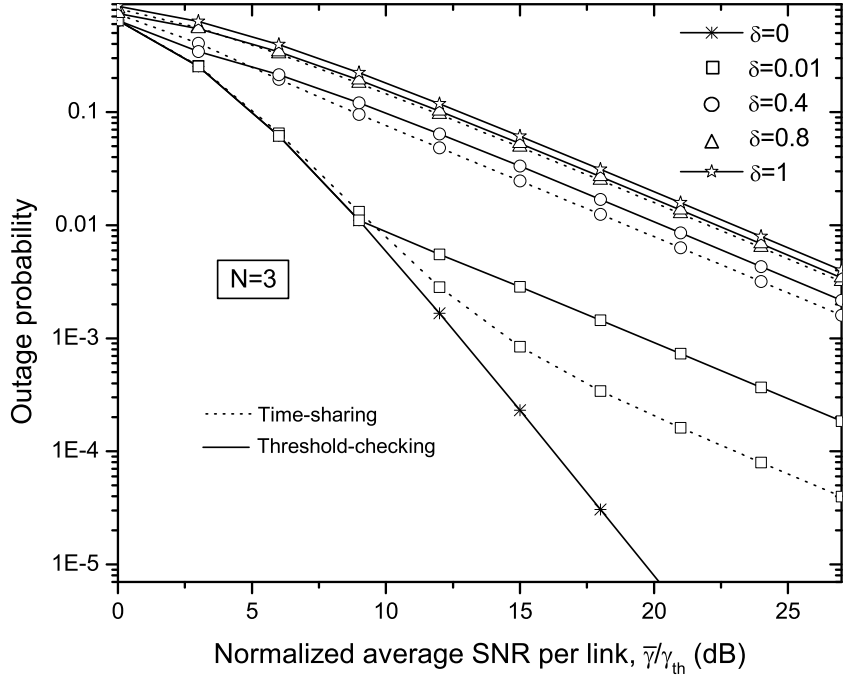}
\caption{Outage probability vs. normalized average SNR per link, $\bar{%
\protect\gamma}/\protect\gamma_{\text{th}}$, for $N=3$ and several
values of the tradeoff factor, $\protect\delta$.}
\label{Fig_OutGN3}
\end{figure}

For the special case of two candidate relays ($N=2$), we developed the
Pareto efficient relay selection policy. This policy yields the optimum
capacity and outage probability for a given energy transfer, as well as the
maximum energy transfer for a given constraint on the capacity or outage
probability. Along with the optimal policy, the selection scheme dubbed
\textquotedblleft weighted difference\textquotedblright\ was also proposed
for $N=2$. This scheme performs similarly to the Pareto efficient scheme,
and yields tractable mathematical analysis. A general conclusion drawn from
our analysis is that the diversity gain is lost when the links transferring
energy to the RF harvester are included in the relay selection decision ($%
\delta >0$), unless the Pareto efficient policy is employed and the links
transferring information to the receiver are also included in the selection
decision ($\delta <1$). Moreover, the Pareto efficient scheme and the
\textquotedblleft weighted difference\textquotedblright\ scheme offer
attractive tradeoffs when operating close to their upper and lower boundary (%
$\delta \approx 0$ and $\delta \approx 1$), in the sense that they achieve
substantial improvement of capacity (and/or outage probability) with
relatively little cost in energy transfer.

\appendices


\section{\label{Cmax}Proof of Lemma \protect\ref{lem1}}

Clearly, the maximum ergodic capacity of the information transmission to $%
\mathcal{D}$ equals the ergodic capacity of the $\mathcal{S}$-$\mathcal{R}%
_{\kappa }$-$\mathcal{D}$ link. For DF relaying, the information rate to $%
\mathcal{D}$ is dominated by the bottleneck link, i.e., by the weakest of
the $\mathcal{S}$-$\mathcal{R}_{\kappa }$ and $\mathcal{R}_{\kappa }$-$%
\mathcal{D}$ links. Thus, the maximum ergodic capacity is obtained as%
\begin{equation}
C_{\max }=\frac{1}{2}\int_{0}^{\infty }\log _{2}\left( 1+\gamma _{\kappa
}\right) f_{\gamma _{\kappa }}\left( x\right) dx  \label{Cerg}
\end{equation}%
where the pre-log factor $1/2$ is used because of the half-duplex assumption
and $f_{\gamma _{\kappa }}\left( \cdot \right) $ denotes the probability
density function (PDF) of $\gamma _{\kappa }=\max_{i=1,...,N}\gamma _{i}$.
Since $\left\{ \gamma _{1},\gamma _{2},...,\gamma _{N}\right\} $ is a set of
exponentially distributed RVs, $f_{\gamma _{\kappa }}\left( \cdot \right) $
can be obtained from the theory of ordered statistics \cite{B:David} and (%
\ref{gi}) as%
\begin{equation}
f_{\gamma _{\kappa }}\left( x\right) =N\sum_{j=0}^{N-1}\left( -1\right) ^{j}%
\binom{N-1}{j}\frac{2}{\bar{\gamma}}\exp \left( -2x\frac{j+1}{\bar{\gamma}}%
\right) .  \label{gkap}
\end{equation}%
Plugging (\ref{gkap}) into (\ref{Cerg}) and using integration by parts
yields (\ref{CN}).

The minimum ergodic capacity occurs in the case where the CSI of the $%
\mathcal{S}$-$\mathcal{R}_{i}$-$\mathcal{D}$ links is not exploited for
relay selection, or equivalently, when the relay is selected based on a
process which is independent of the $\mathcal{S}$-$\mathcal{R}_{i}$-$%
\mathcal{D}$ channel strength\footnote{%
In fact, the ergodic capacity can reach even lower values, if a relay that
is known to have weak channel conditions is selected. This case, however, is
not in line with the concept of opportunistic relay selection, and is
therefore not considered here.}. The minimum ergodic capacity is obtained
directly from (\ref{CN}), by setting $N=1$. In that case, (\ref{CN}) reduces
to (\ref{C1}).

\section{\label{ApCTB}Derivation of (\protect\ref{CTB})}

The first term in (\ref{CTB}) is obtained by using (\ref{gkap}) and
employing integration by parts, yielding%
\begin{align}
& E\left\{ C_{\kappa }\left\vert \gamma _{\kappa }\geq \tau \right. \right\}
\Pr \left\{ \gamma _{\kappa }\geq \tau \right\} =\frac{\int_{\tau }^{\infty }%
\frac{1}{2}\log \left( 1+\gamma _{\kappa }\right) f_{\gamma _{\kappa
}}\left( \gamma _{\kappa }\right) d\gamma _{\kappa }}{\Pr \left\{ \gamma
_{\kappa }\geq \tau \right\} }\Pr \left\{ \gamma _{\kappa }\geq \tau \right\}
\notag \\
& =\sum_{j=0}^{M-1}\frac{N\left( -1\right) ^{j}\binom{N-1}{j}e^{-\frac{%
2\left( j+1\right) \tau }{\bar{\gamma}}}\left[ e^{\frac{2\left( j+1\right)
\left( \tau +1\right) }{\bar{\gamma}}}\mathcal{E}_{1}\left( \frac{2\left(
j+1\right) \left( \tau +1\right) }{\bar{\gamma}}\right) +\ln \left( \tau
+1\right) \right] }{2\left( j+1\right) \ln \left( 2\right) }.  \label{Ap1}
\end{align}%
The second term in (\ref{CTB}) is obtained as%
\begin{equation}
E\left\{ C_{\lambda }\left\vert \gamma _{\kappa }<\tau \right. \right\} \Pr
\left\{ \gamma _{\kappa }<\tau \right\} =\sum_{i=1}^{N}\Pr \left\{
s=i\right\} \frac{\int_{0}^{\tau }\frac{1}{2}\log _{2}\left( 1+x\right)
f_{\gamma _{i}}\left( x\right) dx}{\int_{0}^{\tau }f_{\gamma _{i}}\left(
y\right) dy}\left( 1-e^{-\frac{2\tau }{\bar{\gamma}}}\right) ^{N}.
\label{Ap2}
\end{equation}%
Hence, using (\ref{gi}), (\ref{Ap2}) reduces to%
\begin{equation}
E\left\{ C_{\lambda }\left\vert \gamma _{\kappa }<\tau \right. \right\} \Pr
\left\{ \gamma _{\kappa }<\tau \right\} =\frac{e^{\frac{2}{\bar{\gamma}}}%
\left[ \mathcal{E}_{1}\left( \frac{2}{\bar{\gamma}}\right) -\mathcal{E}%
_{1}\left( \frac{2\left( 1+\tau \right) }{\bar{\gamma}}\right) \right] -e^{-%
\frac{2\tau }{\bar{\gamma}}}\ln \left( 1+\tau \right) }{2\ln \left( 2\right)
}\left( 1-e^{-\frac{2\tau }{\bar{\gamma}}}\right) ^{N-1}.  \label{Ap3}
\end{equation}%
Adding (\ref{Ap1}) and (\ref{Ap3}) yields (\ref{CTB}).

\section{\label{PrProp}Proof of Proposition \protect\ref{Prop1}}

For analyzing the tradeoff of the weighted difference scheme, we distinguish
the following four cases. The selected relay in each case follows directly
from (\ref{WD}).

\begin{itemize}
\item \emph{Case 1: }$\gamma_{1}<\gamma_{2}$ and $\varepsilon_{1}<%
\varepsilon_{2}$. \emph{Selected relay:} $s=2$.

\item \emph{Case 2: }$\gamma_{1}<\gamma_{2}$ and $\varepsilon_{1}>%
\varepsilon_{2}$. \emph{Selected relay:} $s=2$ if $\varepsilon_{1}<%
\varepsilon_{2}+\left( \gamma_{2}-\gamma_{1}\right) /\nu$; $s=1$ if $%
\varepsilon_{1}>\varepsilon_{2}+\left( \gamma_{2}-\gamma_{1}\right) /\nu$.

\item \emph{Case 3: }$\gamma_{1}>\gamma_{2}$ and $\varepsilon_{1}<%
\varepsilon_{2}$. \emph{Selected relay:} $s=1$ if $\varepsilon_{2}<%
\varepsilon_{1}+\left( \gamma_{1}-\gamma_{2}\right) /\nu$; $s=2$ if $%
\varepsilon_{2}>\varepsilon_{1}+\left( \gamma_{1}-\gamma_{2}\right) /\nu$.

\item \emph{Case 4: }$\gamma_{1}>\gamma_{2}$ and $\varepsilon_{1}>%
\varepsilon_{2}$. \emph{Selected relay: }$s=1$.
\end{itemize}

Considering the above cases, we can express the average energy transfer
function of $\nu $ as%
\begin{align}
\epsilon _{WD}& =\underset{\mathcal{I}_{1}\text{ (\emph{Case 1})}}{%
\underbrace{\int_{0}^{\infty }f_{\gamma _{1}}\left( \gamma _{1}\right)
\int_{\gamma _{1}}^{\infty }f_{\gamma _{2}}\left( \gamma _{2}\right)
\int_{0}^{\infty }f_{\varepsilon _{2}}\left( \varepsilon _{2}\right)
\int_{0}^{\varepsilon _{2}}\varepsilon _{2}f_{\varepsilon _{1}}\left(
\varepsilon _{1}\right) d\varepsilon _{1}d\varepsilon _{2}d\gamma
_{2}d\gamma _{1}}}  \notag \\
& +\underset{\mathcal{I}_{2a}\text{ (\emph{Case 2a})}}{\underbrace{%
\int_{0}^{\infty }f_{\gamma _{1}}\left( \gamma _{1}\right) \int_{\gamma
_{1}}^{\infty }f_{\gamma _{2}}\left( \gamma _{2}\right) \int_{0}^{\infty
}f_{\varepsilon _{2}}\left( \varepsilon _{2}\right) \int_{\varepsilon
_{2}}^{\varepsilon _{2}+\left( \gamma _{2}-\gamma _{1}\right) /\nu
}\varepsilon _{2}f_{\varepsilon _{1}}\left( \varepsilon _{1}\right)
d\varepsilon _{1}d\varepsilon _{2}d\gamma _{2}d\gamma _{1}}}  \notag \\
& +\underset{\mathcal{I}_{2b}\text{ (\emph{Case 2b})}}{\underbrace{%
\int_{0}^{\infty }f_{\gamma _{1}}\left( \gamma _{1}\right) \int_{\gamma
_{1}}^{\infty }f_{\gamma _{2}}\left( \gamma _{2}\right) \int_{0}^{\infty
}f_{\varepsilon _{2}}\left( \varepsilon _{2}\right) \int_{\varepsilon
_{2}+\left( \gamma _{2}-\gamma _{1}\right) /\nu }^{\infty }\varepsilon
_{1}f_{\varepsilon _{1}}\left( \varepsilon _{1}\right) d\varepsilon
_{1}d\varepsilon _{2}d\gamma _{2}d\gamma _{1}}}  \notag \\
& +\underset{\mathcal{I}_{3a}\text{ (\emph{Case 3a})}}{\underbrace{%
\int_{0}^{\infty }f_{\gamma _{1}}\left( \gamma _{1}\right) \int_{0}^{\gamma
_{1}}f_{\gamma _{2}}\left( \gamma _{2}\right) \int_{0}^{\infty
}f_{\varepsilon _{1}}\left( \varepsilon _{1}\right) \int_{\varepsilon
_{1}}^{\varepsilon _{1}+\left( \gamma _{1}-\gamma _{2}\right) /\nu
}\varepsilon _{1}f_{\varepsilon _{2}}\left( \varepsilon _{2}\right)
d\varepsilon _{2}d\varepsilon _{1}d\gamma _{2}d\gamma _{1}}}  \notag \\
& +\underset{\mathcal{I}_{3b}\text{ (\emph{Case 3b})}}{\underbrace{%
\int_{0}^{\infty }f_{\gamma _{1}}\left( \gamma _{1}\right) \int_{0}^{\gamma
_{1}}f_{\gamma _{2}}\left( \gamma _{2}\right) \int_{0}^{\infty
}f_{\varepsilon _{1}}\left( \varepsilon _{1}\right) \int_{\varepsilon
_{1}+\left( \gamma _{1}-\gamma _{2}\right) /\nu }^{\infty }\varepsilon
_{2}f_{\varepsilon _{2}}\left( \varepsilon _{2}\right) d\varepsilon
_{2}d\varepsilon _{1}d\gamma _{2}d\gamma _{1}}}  \notag \\
& +\underset{\mathcal{I}_{4}\text{ (\emph{Case 4})}}{\underbrace{%
\int_{0}^{\infty }f_{\gamma _{1}}\left( \gamma _{1}\right) \int_{0}^{\gamma
_{1}}f_{\gamma _{2}}\left( \gamma _{2}\right) \int_{0}^{\infty
}f_{\varepsilon _{2}}\left( \varepsilon _{2}\right) \int_{\varepsilon
_{2}}^{\infty }\varepsilon _{1}f_{\varepsilon _{1}}\left( \varepsilon
_{1}\right) d\varepsilon _{1}d\varepsilon _{2}d\gamma _{2}d\gamma _{1}}}.
\label{eWD}
\end{align}%
Using elementary integrations and \cite[Eq. (3.351.7)]{B:Gra_Ryz_Book}, (\ref%
{eWD}) reduces after algebraic manipulations to%
\begin{align}
\epsilon _{WD}& =\underset{\mathcal{I}_{1}}{\underbrace{\frac{3}{8}\bar{%
\varepsilon}}}+\underset{\mathcal{I}_{2a}}{\underbrace{\frac{\bar{\varepsilon%
}~\bar{\gamma}}{8\bar{\gamma}+16\nu \bar{\varepsilon}}}}+\underset{\mathcal{I%
}_{2b}}{\underbrace{\frac{\bar{\varepsilon}^{2}\nu \left( 5\bar{\gamma}+6\nu
\bar{\varepsilon}\right) }{4\left( \bar{\gamma}+2\nu \bar{\varepsilon}%
\right) ^{2}}}}+\underset{\mathcal{I}_{3a}}{\underbrace{\frac{\bar{%
\varepsilon}~\bar{\gamma}}{8\bar{\gamma}+16\nu \bar{\varepsilon}}}}+\underset%
{\mathcal{I}_{3b}}{\underbrace{\frac{\bar{\varepsilon}^{2}\nu \left( 5\bar{%
\gamma}+6\nu \bar{\varepsilon}\right) }{4\left( \bar{\gamma}+2\nu \bar{%
\varepsilon}\right) ^{2}}}}+\underset{\mathcal{I}_{4}}{\underbrace{\frac{3}{8%
}\bar{\varepsilon}}}  \notag \\
& =\frac{\bar{\varepsilon}}{2}\left[ 3-\frac{\bar{\gamma}^{2}}{\left( \bar{%
\gamma}+2\nu \bar{\varepsilon}\right) ^{2}}\right] .  \label{eWD2}
\end{align}%
It is interesting to observe from (\ref{eWD2}) the following: $\mathcal{I}%
_{1}=\mathcal{I}_{4}$; $\mathcal{I}_{2a}=\mathcal{I}_{3a}$; $\mathcal{I}%
_{2b}=\mathcal{I}_{3b}$; $\mathcal{I}_{1}$ and $\mathcal{I}_{4}$ are
independent of $\bar{\gamma}$. All the above observations are explained by
the assumption that all participating channels are i.i.d. Solving (\ref{eWD2}%
) with respect to $\nu $ yields%
\begin{equation}
\nu =\frac{\bar{\gamma}}{2\bar{\varepsilon}}\left( \sqrt{\frac{\bar{%
\varepsilon}}{3\bar{\varepsilon}-2\epsilon _{WD}}}-1\right) ,~\bar{%
\varepsilon}\leq \epsilon _{WD}\leq H_{2}\bar{\varepsilon}.  \label{nu}
\end{equation}

The ergodic capacity of the weighted difference scheme is calculated in a
way similar to the average transferred energy. Considering again \emph{Case
1 - Case 4}, yields%
\begin{align}
C_{WD}& =\underset{\mathcal{J}_{1}\text{ (\emph{Case 1})}}{\underbrace{%
\int_{0}^{\infty }f_{\gamma _{1}}\left( \gamma _{1}\right) \int_{\gamma
_{1}}^{\infty }f_{\gamma _{2}}\left( \gamma _{2}\right) \int_{0}^{\infty
}f_{\varepsilon _{2}}\left( \varepsilon _{2}\right) \int_{0}^{\varepsilon
_{2}}\frac{\log _{2}\left( 1+\gamma _{2}\right) }{2}f_{\varepsilon
_{1}}\left( \varepsilon _{1}\right) d\varepsilon _{1}d\varepsilon
_{2}d\gamma _{2}d\gamma _{1}}}  \notag \\
& +\underset{\mathcal{J}_{2a}\text{ (\emph{Case 2a})}}{\underbrace{%
\int_{0}^{\infty }f_{\gamma _{1}}\left( \gamma _{1}\right) \int_{\gamma
_{1}}^{\infty }f_{\gamma _{2}}\left( \gamma _{2}\right) \int_{0}^{\infty
}f_{\varepsilon _{2}}\left( \varepsilon _{2}\right) \int_{\varepsilon
_{2}}^{\varepsilon _{2}+\left( \gamma _{2}-\gamma _{1}\right) /\nu }\frac{%
\log _{2}\left( 1+\gamma _{2}\right) }{2}f_{\varepsilon _{1}}\left(
\varepsilon _{1}\right) d\varepsilon _{1}d\varepsilon _{2}d\gamma
_{2}d\gamma _{1}}}  \notag \\
& +\underset{\mathcal{J}_{2b}\text{ (\emph{Case 2b})}}{\underbrace{%
\int_{0}^{\infty }f_{\gamma _{1}}\left( \gamma _{1}\right) \int_{\gamma
_{1}}^{\infty }f_{\gamma _{2}}\left( \gamma _{2}\right) \int_{0}^{\infty
}f_{\varepsilon _{2}}\left( \varepsilon _{2}\right) \int_{\varepsilon
_{2}+\left( \gamma _{2}-\gamma _{1}\right) /\nu }^{\infty }\frac{\log
_{2}\left( 1+\gamma _{1}\right) }{2}f_{\varepsilon _{1}}\left( \varepsilon
_{1}\right) d\varepsilon _{1}d\varepsilon _{2}d\gamma _{2}d\gamma _{1}}}
\notag \\
& +\underset{\mathcal{J}_{3a}\text{ (\emph{Case 3a})}}{\underbrace{%
\int_{0}^{\infty }f_{\gamma _{1}}\left( \gamma _{1}\right) \int_{0}^{\gamma
_{1}}f_{\gamma _{2}}\left( \gamma _{2}\right) \int_{0}^{\infty
}f_{\varepsilon _{1}}\left( \varepsilon _{1}\right) \int_{\varepsilon
_{1}}^{\varepsilon _{1}+\left( \gamma _{1}-\gamma _{2}\right) /\nu }\frac{%
\log _{2}\left( 1+\gamma _{1}\right) }{2}f_{\varepsilon _{2}}\left(
\varepsilon _{2}\right) d\varepsilon _{2}d\varepsilon _{1}d\gamma
_{2}d\gamma _{1}}}  \notag \\
& +\underset{\mathcal{J}_{3b}\text{ (\emph{Case 3b})}}{\underbrace{%
\int_{0}^{\infty }f_{\gamma _{1}}\left( \gamma _{1}\right) \int_{0}^{\gamma
_{1}}f_{\gamma _{2}}\left( \gamma _{2}\right) \int_{0}^{\infty
}f_{\varepsilon _{1}}\left( \varepsilon _{1}\right) \int_{\varepsilon
_{1}+\left( \gamma _{1}-\gamma _{2}\right) /\nu }^{\infty }\frac{\log
_{2}\left( 1+\gamma _{2}\right) }{2}f_{\varepsilon _{2}}\left( \varepsilon
_{2}\right) d\varepsilon _{2}d\varepsilon _{1}d\gamma _{2}d\gamma _{1}}}
\notag \\
& \underset{\mathcal{J}_{4}\text{ (\emph{Case 4})}}{+\underbrace{%
\int_{0}^{\infty }f_{\gamma _{1}}\left( \gamma _{1}\right) \int_{0}^{\gamma
_{1}}f_{\gamma _{2}}\left( \gamma _{2}\right) \int_{0}^{\infty
}f_{\varepsilon _{2}}\left( \varepsilon _{2}\right) \int_{\varepsilon
_{2}}^{\infty }\frac{\log _{2}\left( 1+\gamma _{1}\right) }{2}f_{\varepsilon
_{1}}\left( \varepsilon _{1}\right) d\varepsilon _{1}d\varepsilon
_{2}d\gamma _{2}d\gamma _{1}}}.  \label{CWD}
\end{align}%
Using integration by parts where appropriate, in conjunction with \cite[Eq.
(3.351.7)]{B:Gra_Ryz_Book} and \cite[Eq. (4.337.2)]{B:Gra_Ryz_Book}, (\ref%
{CWD}) reduces after algebraic manipulations to%
\begin{align}
C_{WD}& =\underset{\mathcal{J}_{1}~\left( \mathcal{J}_{4}\right) }{2~%
\underbrace{\frac{2e^{\frac{2}{\bar{\gamma}}}\mathcal{E}_{1}\left( \frac{2}{%
\bar{\gamma}}\right) -e^{\frac{4}{\bar{\gamma}}}\mathcal{E}_{1}\left( \frac{4%
}{\bar{\gamma}}\right) }{8\ln \left( 2\right) }}}  \notag \\
& +2~\underset{\mathcal{J}_{2}~\left( \mathcal{J}_{3}\right) }{\underbrace{%
\frac{2\left( \bar{\gamma}^{2}-4\nu ^{2}\bar{\varepsilon}^{2}\right)
\mathcal{E}_{1}\left( \frac{2}{\bar{\gamma}}\right) +8e^{\frac{1}{\nu \bar{%
\varepsilon}}}\nu ^{2}\bar{\varepsilon}^{2}\mathcal{E}_{1}\left( \frac{2}{%
\bar{\gamma}}+\frac{1}{\nu \bar{\varepsilon}}\right) -e^{\frac{2}{\bar{\gamma%
}}}\left( \bar{\gamma}^{2}+4\nu ^{2}\bar{\varepsilon}^{2}\right) \mathcal{E}%
_{1}\left( \frac{4}{\bar{\gamma}}\right) }{8e^{-\frac{2}{\bar{\gamma}}%
}\left( \bar{\gamma}^{2}-4\nu ^{2}\bar{\varepsilon}^{2}\right) \ln \left(
2\right) }}}  \label{CWD2a} \\
& =\frac{2\left( \bar{\gamma}^{2}-4\nu ^{2}\bar{\varepsilon}^{2}\right)
\mathcal{E}_{1}\left( \frac{2}{\bar{\gamma}}\right) +4e^{\frac{1}{\nu \bar{%
\varepsilon}}}\nu ^{2}\bar{\varepsilon}^{2}\mathcal{E}_{1}\left( \frac{2}{%
\bar{\gamma}}+\frac{1}{\nu \bar{\varepsilon}}\right) -e^{\frac{2}{\bar{\gamma%
}}}\bar{\gamma}^{2}\mathcal{E}_{1}\left( \frac{4}{\bar{\gamma}}\right) }{%
2e^{-\frac{2}{\bar{\gamma}}}\left( \bar{\gamma}^{2}-4\nu ^{2}\bar{\varepsilon%
}^{2}\right) \ln \left( 2\right) }  \label{CWD2b}
\end{align}%
where the factor $2$ in front of each of the two terms in (\ref{CWD2a}) is
due to symmetry, similar to the observations in (\ref{eWD2}). Plugging (\ref%
{nu}) into (\ref{CWD2b}) yields (\ref{WDTrad}).

\section{\label{AppOut}Proof of Proposition \protect\ref{Theo2}}

The outage probability of the weighted difference scheme is obtained by
utilizing the four cases considered in Appendix \ref{PrProp}, as follows%
\begin{align}
P_{out,WD} & =\underset{\mathcal{K}_{1}\text{ (\emph{Case 1})}}{\underbrace{%
\int_{0}^{\gamma_{\text{th}}}f_{\gamma_{2}}\left( \gamma_{2}\right)
\int_{0}^{\gamma_{2}}f_{\gamma_{1}}\left( \gamma_{1}\right)
\int_{0}^{\infty}f_{\varepsilon_{2}}\left( \varepsilon_{2}\right) \int
_{0}^{\varepsilon_{2}}f_{\varepsilon_{1}}\left( \varepsilon_{1}\right)
d\varepsilon_{1}d\varepsilon_{2}d\gamma_{1}d\gamma_{2}}}  \notag \\
& +\underset{\mathcal{K}_{2a}\text{ (\emph{Case 2a})}}{\underbrace{\int
_{0}^{\gamma_{\text{th}}}f_{\gamma_{2}}\left( \gamma_{2}\right)
\int_{0}^{\gamma_{2}}f_{\gamma_{1}}\left( \gamma_{1}\right) \int_{0}^{\infty
}f_{\varepsilon_{2}}\left( \varepsilon_{2}\right)
\int_{\varepsilon_{2}}^{\varepsilon_{2}+\left( \gamma_{2}-\gamma_{1}\right)
/\nu}f_{\varepsilon _{1}}\left( \varepsilon_{1}\right)
d\varepsilon_{1}d\varepsilon_{2}d\gamma_{1}d\gamma_{2}}}  \notag \\
& +\underset{\mathcal{K}_{2b}\text{ (\emph{Case 2b})}}{\underbrace{\int
_{0}^{\gamma_{\text{th}}}f_{\gamma_{1}}\left( \gamma_{1}\right)
\int_{\gamma_{1}}^{\infty}f_{\gamma_{2}}\left( \gamma_{2}\right)
\int_{0}^{\infty }f_{\varepsilon_{2}}\left( \varepsilon_{2}\right)
\int_{\varepsilon _{2}+\left( \gamma_{2}-\gamma_{1}\right)
/\nu}^{\infty}f_{\varepsilon_{1}}\left( \varepsilon_{1}\right)
d\varepsilon_{1}d\varepsilon_{2}d\gamma _{2}d\gamma_{1}}}  \notag \\
& +\underset{\mathcal{K}_{3a}\text{ (\emph{Case 3a})}}{\underbrace{\int
_{0}^{\gamma_{\text{th}}}f_{\gamma_{1}}\left( \gamma_{1}\right)
\int_{0}^{\gamma_{1}}f_{\gamma_{2}}\left( \gamma_{2}\right) \int_{0}^{\infty
}f_{\varepsilon_{1}}\left( \varepsilon_{1}\right)
\int_{\varepsilon_{1}}^{\varepsilon_{1}+\left( \gamma_{1}-\gamma_{2}\right)
/\nu}f_{\varepsilon _{2}}\left( \varepsilon_{2}\right)
d\varepsilon_{2}d\varepsilon_{1}d\gamma_{2}d\gamma_{1}}}  \notag \\
& +\underset{\mathcal{K}_{3b}\text{ (\emph{Case 3b})}}{\underbrace{\int
_{0}^{\gamma_{\text{th}}}f_{\gamma_{2}}\left( \gamma_{2}\right)
\int_{\gamma_{2}}^{\infty}f_{\gamma_{1}}\left( \gamma_{1}\right)
\int_{0}^{\infty }f_{\varepsilon_{1}}\left( \varepsilon_{1}\right)
\int_{\varepsilon _{1}+\left( \gamma_{1}-\gamma_{2}\right)
/\nu}^{\infty}f_{\varepsilon_{2}}\left( \varepsilon_{2}\right)
d\varepsilon_{2}d\varepsilon_{1}d\gamma _{1}d\gamma_{2}}}  \notag \\
& \underset{\mathcal{K}_{4}\text{ (\emph{Case 4})}}{+\underbrace{\int
_{0}^{\gamma_{\text{th}}}f_{\gamma_{1}}\left( \gamma_{1}\right)
\int_{0}^{\gamma_{1}}f_{\gamma_{2}}\left( \gamma_{2}\right) \int_{0}^{\infty
}f_{\varepsilon_{2}}\left( \varepsilon_{2}\right)
\int_{\varepsilon_{2}}^{\infty}f_{\varepsilon_{1}}\left(
\varepsilon_{1}\right) d\varepsilon
_{1}d\varepsilon_{2}d\gamma_{2}d\gamma_{1}}}.  \label{OutWD1}
\end{align}
Working similarly as in Appendix \ref{PrProp}, we can simplify (\ref{OutWD1}%
) to%
\begin{align}
P_{out,WD} & =\underset{\mathcal{K}_{1}\left( \mathcal{K}_{4}\right) }{~2~%
\underbrace{\frac{1}{4}e^{-\frac{4\gamma_{\text{th}}}{\bar{\gamma}}}\left(
e^{\frac{2\gamma_{\text{th}}}{\bar{\gamma}}}-1\right) ^{2}}}+~2~\underset{%
\mathcal{K}_{2}\left( \mathcal{K}_{3}\right) }{\underbrace{\left\{ \frac{1}{4%
}1+e^{-\frac{4\gamma_{\text{th}}}{\bar{\gamma}}}\left[ \frac{\bar{\gamma }%
^{2}+4\nu^{2}\bar{\varepsilon}^{2}}{\bar{\gamma}^{2}-4\nu^{2}\bar {%
\varepsilon}^{2}}-2e^{\frac{2\gamma_{\text{th}}}{\bar{\gamma}}}\left( 1+%
\frac{4e^{-\frac{\gamma_{\text{th}}}{\nu\bar{\varepsilon}}}\nu^{2}\bar{%
\varepsilon }^{2}}{\bar{\gamma}^{2}-4\nu^{2}\bar{\varepsilon}^{2}}\right) %
\right] \right\} }}  \notag \\
& =\frac{e^{-\frac{4\gamma_{\text{th}}}{\bar{\gamma}}}\left( e^{\frac{%
2\gamma_{\text{th}}}{\bar{\gamma}}}-1\right) ^{2}\bar{\gamma}^{2}+4\nu^{2}%
\bar{\varepsilon}^{2}\left[ e^{-\frac{2\gamma_{\text{th}}}{\bar{\gamma}}%
}\left( 2-e^{-\frac {\gamma_{\text{th}}}{\nu\bar{\varepsilon}}}\right) -1%
\right] }{\bar{\gamma}^{2}-4\nu^{2}\bar{\varepsilon}^{2}}.  \label{OutWD2}
\end{align}
Substituting (\ref{nu}) in (\ref{OutWD2}) yields the outage probability of
the weighted difference scheme expressed as a function of $\epsilon_{WD}$
and $\bar{\varepsilon}$, as shown in (\ref{OutWD0}). Using (\ref{delta}), (%
\ref{OutWD02}) is derived from (\ref{OutWD0}).

\section{\label{ProofTheo}Proof of Theorem \protect\ref{Theo}}

By its definition, $\mathcal{F}$ represents the average of a performance
metric over a window of $M$ transmission sessions, when $M\rightarrow \infty
$. In order to mathematically express the selection of either $\mathcal{R}%
_{1}$ or $\mathcal{R}_{2}$ in a given transmission session, $m$, we
introduce the binary auxiliary variable $w_{m}$, such that $w_{m}=1$ if $s=1$%
; $w_{m}=0$ if $s=2$. Then, the problem of maximizing $\mathcal{F}$ for
given energy transfer constraints is expressed as%
\begin{align}
& \max_{w_{m}}\lim_{M\rightarrow \infty }\frac{1}{M}\sum_{m=1}^{M}\left[
w_{m}\mathcal{F}\left( \gamma _{1,m}\right) +\left( 1-w_{m}\right) \mathcal{F%
}\left( \gamma _{2,m}\right) \right]  \notag \\
& \text{s.t. ~~\ ~\ \ ~~}\frac{1}{M}\sum_{m=1}^{M}w_{m}\left( 1-w_{m}\right)
=0  \notag \\
~~~~~& ~~~~~~~~~~~~\lim_{M\rightarrow \infty }\ \frac{1}{M}\sum_{m=1}^{M}%
\left[ w_{m}\varepsilon _{1,m}+\left( 1-w_{m}\right) \varepsilon _{2,m}%
\right] \geq \epsilon  \label{Gen1}
\end{align}%
where $\left\{ \gamma _{1,m},\varepsilon _{1,m}\right\} $, $\left\{ \gamma
_{2,m},\varepsilon _{2,m}\right\} $ denote the \{SNR, harvested energy\} of
the $\mathcal{S}$-$\mathcal{R}_{1}$-$\mathcal{D}$ and $\mathcal{S}$-$%
\mathcal{R}_{2}$-$\mathcal{D}$ links, respectively, in transmission frame $m$%
. Using the parameters $\xi _{m}$ and $\zeta $ as non-negative Langrange
multipliers, the Langrangian of the above problem is obtained as%
\begin{align}
\mathcal{L}& =\lim_{M\rightarrow \infty }\frac{1}{M}\sum_{m=1}^{M}\left[
w_{m}\mathcal{F}\left( \gamma _{1,m}\right) +\left( 1-w_{m}\right) \mathcal{F%
}\left( \gamma _{2,m}\right) \right]  \notag \\
& +\frac{\xi _{m}}{M}\sum_{m=1}^{M}w_{m}\left( 1-w_{m}\right)
+\lim_{M\rightarrow \infty }\frac{\zeta }{M}\sum_{m=1}^{M}\left[
w_{m}\varepsilon _{1,m}+\left( 1-w_{m}\right) \varepsilon _{2,m}-\epsilon %
\right] .  \label{L1}
\end{align}%
Let us not concentrate on maximizing $\mathcal{L}$ for a given transmission
frame $m$, and let us drop the index $m$ for notational simplicity. The
derivative of $\mathcal{L}$ with respect to $w$ is obtained from (\ref{L1})
as%
\begin{equation}
\frac{\partial \mathcal{L}}{\partial w}=\frac{1}{M}\mathcal{F}\left( \gamma
_{1}\right) -\frac{1}{M}\mathcal{F}\left( \gamma _{2}\right) +\frac{\xi }{M}%
\left( 1-2w\right) +\frac{\zeta }{M}\left( \varepsilon _{1}-\varepsilon
_{2}\right) .  \label{L2}
\end{equation}%
Setting the derivative in (\ref{L2}) equal to zero and solving with respect
to $w$ yields%
\begin{equation}
w=\frac{\Delta \mathcal{F}+\zeta ~\Delta \varepsilon +\xi }{2\xi }
\label{wm1}
\end{equation}%
where $\Delta $ denotes difference, such that $\Delta \mathcal{F}=\mathcal{F}%
\left( \gamma _{1}\right) -\mathcal{F}\left( \gamma _{2}\right) $ and $%
\Delta \varepsilon =\varepsilon _{1}-\varepsilon _{2}$. Since $w$ is a
binary variable (i.e., it equals either zero or one), (\ref{wm1}) yields%
\begin{equation}
w=\left\{
\begin{array}{c}
0,~~\text{if }\xi =-\Delta \mathcal{F}-\zeta ~\Delta \varepsilon \\
1,~~~\text{if }\xi =\Delta \mathcal{F}+\zeta ~\Delta \varepsilon%
\end{array}%
.\right.  \label{wm2}
\end{equation}%
Considering that $\xi \geq 0$, (\ref{wm2}) yields the optimal relay
selection rule given the value of $\zeta \geq 0$, as follows%
\begin{equation}
w=\left\{
\begin{array}{c}
0,~~\text{if }\Delta \mathcal{F}+\zeta ~\Delta \varepsilon <0 \\
1,~~\text{if }\Delta \mathcal{F}+\zeta ~\Delta \varepsilon \geq 0%
\end{array}%
\right. .  \label{wm3}
\end{equation}%
which is equivalent to (\ref{PO}). Since $\mathcal{F}\left( \gamma
_{i}\right) $ is a non-decreasing function of $\gamma _{i}$, the policy in (%
\ref{PO}) also maximizes the average energy transfer for a given $\mathcal{F}
$. This completes the proof.

\bibliographystyle{IEEEtran}
\bibliography{acompat,References}

\newif\ifabfull\abfulltrue
\begin{thebibliography}{10}
\providecommand{\url}[1]{#1}
\csname url@samestyle\endcsname
\providecommand{\newblock}{\relax}
\providecommand{\bibinfo}[2]{#2}
\providecommand{\BIBentrySTDinterwordspacing}{\spaceskip=0pt\relax}
\providecommand{\BIBentryALTinterwordstretchfactor}{4}
\providecommand{\BIBentryALTinterwordspacing}{\spaceskip=\fontdimen2\font plus
\BIBentryALTinterwordstretchfactor\fontdimen3\font minus
  \fontdimen4\font\relax}
\providecommand{\BIBforeignlanguage}[2]{{%
\expandafter\ifx\csname l@#1\endcsname\relax
\typeout{** WARNING: IEEEtran.bst: No hyphenation pattern has been}%
\typeout{** loaded for the language `#1'. Using the pattern for}%
\typeout{** the default language instead.}%
\else
\language=\csname l@#1\endcsname
\fi
#2}}
\providecommand{\BIBdecl}{\relax}
\BIBdecl

\bibitem{C:FettweisICT}
G.~Fettweis and E.~Zimmermann, ``{ICT} energy consumption, trends and
  challenges,'' in \emph{in Proc. of International Symposium on Wireless
  Personal Multimedia Communications}, 2008.

\bibitem{ST:Cisco}
``Cisco visual networking index: {F}orecast and methodology, 2011-2016, white
  paper,'' [Online]. Available:
  http://www.cisco.com/en/US/solutions/collateral/ns341/ns525/ns537/ns705/ns827/white\_paper\_c11-481360.pdf.

\bibitem{B:Hoss_Green}
E.~Hossain, V.~K. Bhargava, and G.~P. Fettweis, \emph{Green Radio Communication
  Networks}.\hskip 1em plus 0.5em minus 0.4em\relax New York: Cambridge
  University Press, 2012.

\bibitem{J:Zia}
Z.~Hasan, H.~Boostanimehr, and V.~K. Bhargava, ``Green cellular networks: {A}
  survey, some research issues and challenges,'' \emph{IEEE Communications
  Surveys \& Tutorials}, pp. 524--540, Fourth Quarter 2011.

\bibitem{J:Medepally}
B.~Medepally and N.~Mehta, ``Voluntary energy harvesting relays and selection
  in cooperative wireless networks,'' \emph{IEEE Trans Wireless Commun.}, pp.
  3543--3553, Nov. 2010.

\bibitem{J:Ho1}
C.~K. Ho and R.~Zhang, ``Optimal energy allocation for wireless communications
  with energy harvesting constraints,'' \emph{IEEE Trans. Signal Process.}, pp.
  4808--4818, Sep. 2012.

\bibitem{J:Ozel}
O.~Ozel, K.~Tutuncuoglu, J.~Yang, S.~Ulukus, and A.~Yener, ``Transmission with
  energy harvesting nodes in fading wireless channels: {O}ptimal policies,''
  \emph{IEEE J. Sel. Areas Commun.}, vol.~29, pp. 1732--1743, Sep. 2011.

\bibitem{J:Sudeva}
S.~Sudevalayam and P.~Kulkarni, ``Energy harvesting sensor nodes: {S}urvey and
  implications,'' \emph{IEEE Communications Surveys \& Tutorials}, pp.
  443--461, Third Quarter 2011.

\bibitem{J:Nasir}
A.~A. Nasir, X.~Zhou, S.~Durrani, and R.~A. Kennedy, ``Relaying protocols for
  wireless energy harvesting and information processing,'' 2012, submitted for
  possible publication in \textit{IEEE Trans. Wireless Commun.} [Online].
  Available: http://arxiv.org/abs/1212.5406.

\bibitem{J:HuangLau}
K.~Huang and V.~K.~N. Lau, ``Enabling wireless power transfer in cellular
  networks: {A}rchitecture, modelling and deployment,'' 2012, submitted for
  possible publication in \textit{IEEE Trans. Wireless Commun.} [Online].
  Available: http://arxiv.org/abs/1105.4999.

\bibitem{ST:Pcast}
``Powercast chipset and {RF} energy harvesting reference design enable low-cost
  wireless power over distance,'' Powercast Corporation, 2012. [Online]
  Available:
  http://www.powercastco.com/powercasts-chipset-and-rf-energy-harvesting-reference-design-enable-low-cost-wireless-power-over-distance-20120222/.

\bibitem{J:KrikidisRF}
I.~Krikidis, S.~Timotheou, and S.~Sasaki, ``{RF} energy transfer for
  cooperative networks: {D}ata relaying or energy harvesting?'' \emph{IEEE
  Commun. Lett.}, vol.~16, pp. 1772--1775, Nov. 2012.

\bibitem{C:Varshney}
L.~R. Varshney, ``Transporting information and energy simultaneously,'' in
  \emph{Proc. of IEEE International Symposium on Information Theory (ISIT)},
  Toronto, {C}anada, Jul. 2008, pp. 1612--1616.

\bibitem{C:ShanTes}
P.~Grover and A.~Sahai, ``Shannon meets {T}esla: {W}ireless information and
  power transfer,'' in \emph{Proc. of IEEE International Symposium on
  Information Theory (ISIT)}, Austin, {TX}, Jun. 2010, pp. 2363--2367.

\bibitem{C:Shen}
C.~Shen, W.~C. Li, and T.~H. Chang, ``Simultaneous information and energy
  transfer: {A} two-user {MISO} interference channel case,'' in \emph{Proc. of
  IEEE Global Communications Conference (Globecom)}, Anaheim, {CA}, Dec. 2012,
  pp. 3886--3891.

\bibitem{J:ZhangMIMO}
R.~Zhang and C.~K. Ho, ``{MIMO} broadcasting for simultaneous wireless
  information and power transfer,'' \emph{IEEE Trans. Wireless Commun.}, 2012,
  accepted for publication. [Online]. Available:
  http://arxiv.org/abs/1105.4999.

\bibitem{C:Chalise}
B.~K. Chalise, Y.~D. Zhang, and M.~G. Amin, ``Energy harvesting in an {OSTBC}
  based amplify-and-forward relay system,'' in \emph{Proc. of IEEE
  International Conference on Acoustics, Speech and Signal Processing
  (ICASSP)}, Kyoto, {J}apan, Mar. 2012, pp. 3201--3204.

\bibitem{C:IshibashiTarokh}
K.~Ishibashi, H.~Ochiai, and V.~Tarokh, ``Energy harvesting cooperative
  communications,'' in \emph{Proc. of IEEE International Symposium on Personal,
  Indoor, and Mobile Radio Communications (PIMRC)}, Sydney, {A}ustralia, Sep.
  2012, pp. 1819--1823.

\bibitem{B:Mischa_Coop}
M.~Dohler, \emph{Cooperative Communications: Hardware, Channel and PHY},
  1st~ed.\hskip 1em plus 0.5em minus 0.4em\relax New York, NY: Willey, 2010.

\bibitem{B:Fitzek_Coop}
F.~H. Fitzek and E.~M.~D.~Katz, \emph{Cooperation in Wireless Networks:
  Principles and Applications}, 1st~ed.\hskip 1em plus 0.5em minus 0.4em\relax
  Dordrecht, Netherlands: Springer, 2007.

\bibitem{J:Behrouz1}
B.~Maham, A.~Behnad, and M.~Debbah, ``Analysis of outage probability and
  throughput for half-duplex hybrid-{ARQ} relay channels,'' \emph{IEEE Trans.
  Veh. Technol.}, vol.~61, pp. 3061 --3070, Sep. 2012.

\bibitem{B:Starr}
R.~M. Starr, \emph{General Equilibrium Theory: An Introduction}.\hskip 1em plus
  0.5em minus 0.4em\relax Cambridge University Press, 1997.

\bibitem{B:Abr_Ste_Book}
M.~Abramovitz and I.~A. Stegun, \emph{Handbook of Mathematical Functions with
  Formulas, Graphs, and Mathematical Tables}, 9th~ed.\hskip 1em plus 0.5em
  minus 0.4em\relax New York: Dover, 1972.

\bibitem{B:David}
H.~A. David and H.~N. Nagaraja, \emph{Order Statistics}, 3rd~ed.\hskip 1em plus
  0.5em minus 0.4em\relax New York: Wiley, 2003.

\bibitem{J:Giannakis_Param_Perf}
Z.~Wang and G.~B. Giannakis, ``A simple and general parameterization
  quantifying performance in fading channels,'' \emph{IEEE Trans. Commun.},
  vol.~51, pp. 1389--1398, Aug. 2003.

\bibitem{B:Gra_Ryz_Book}
I.~S. Gradshteyn and I.~M. Ryzhik, \emph{Table of Integrals, Series, and
  Products}, 6th~ed.\hskip 1em plus 0.5em minus 0.4em\relax New York: Academic,
  2000.

\end{thebibliography}

\end{document}